\documentclass[a4paper]{amsart}
\usepackage{geometry}
\usepackage[utf8]{inputenc}
\usepackage[T1]{fontenc}
\usepackage{lmodern}
\usepackage{latexsym}
\usepackage{amsmath}
\usepackage{amsfonts}
\usepackage{amssymb}
\usepackage{amsthm}
\usepackage[color=Olive]{todonotes} 
\usepackage{pgfplots,tikz}
\pgfplotsset{compat=1.15}
\usepackage{xspace} %
\usepackage{enumerate}%
\usetikzlibrary{arrows}
\usetikzlibrary{patterns}
\usetikzlibrary{arrows}
\usetikzlibrary{patterns}
\usetikzlibrary{shapes,decorations}
\usetikzlibrary{decorations.markings}
\usetikzlibrary{calc}
\usepackage{lineno}

\theoremstyle{plain}
\newtheorem{theorem}{Theorem}
\newtheorem{definition}{Definition}
\newtheorem{lemma}[theorem]{Lemma}
\newtheorem{claim}{Claim}[theorem]
\newtheorem{proposition}{Proposition}
\newtheorem{problem}[theorem]{Problem}

\let\leq\leqslant
\let\geq\geqslant

\def\calG{\mathcal{G}}
\def\calF{\mathcal{F}}

\newcommand{\NP}{\ensuremath{\mathsf{NP}}\xspace}

\newcommand{\Oh}[1]{\ensuremath{\mathcal{O}(#1)}}
\newcommand{\yes}{{\bf{yes}}}
\newcommand{\no}{{\bf{no}}}
\newcommand{\bpd}{{\sc{Bipartite Permutation Vertex Deletion }}}
\newcommand{\FPT}{\ensuremath{\mathsf{FPT}}\xspace}

\newcommand{\W}[1]{\ensuremath{\mathsf{W[#1]}}\xspace}

\author[\L{}. Bo\.zyk, J. Derbisz, T. Krawczyk, J. Novotn\'a, K. Okrasa]{\L{}ukasz Bo\.zyk$^{1\dag}$, Jan Derbisz$^2$, Tomasz Krawczyk$^{3\ddag}$, \\Jana Novotn\'a$^{4\dag}$, and Karolina Okrasa$^{5\dag}$}

\address{$^1$Faculty of Mathematics, Informatics and Mechanics, University of Warsaw, Poland}
\email{l.bozyk@uw.edu.pl}

\address{$^2$Theoretical Computer Science Department, 
Faculty of Mathematics and Computer Science, Jagiellonian University in Krak\'ow, Poland.}
\email{jan.derbisz@doctoral.uj.edu.pl}

\address{$^3$Theoretical Computer Science Department, 
Faculty of Mathematics and Computer Science, Jagiellonian University in Krak\'ow, Poland.}
\email{krawczyk@tcs.uj.edu.pl}

\address{$^4$Faculty of Mathematics, Informatics and Mechanics, University of Warsaw, Poland \break Faculty of Mathematics and Physics, Charles University, Czech Republic }
\email{janca@kam.mff.cuni.cz}

\address{$^5$Faculty of Mathematics, Informatics and Mechanics, University of Warsaw, Poland\break Faculty of Mathematics and Information Science, Warsaw University of Technology,\break Poland}
\email{k.okrasa@mini.pw.edu.pl}

\thanks{$^*$The extended abstract of the work was accepted to the conference IPEC 2020.}

\thanks{$^\dag$This research is part of a project that has received funding from the European Research Council (ERC) under the European Union's Horizon 2020 research and innovation programme, grant agreement 714704.
It was partially carried out during the Parameterized Algorithms Retreat of the University of Warsaw, PARUW 2020, held in Krynica-Zdr\'oj in February 2020.}

\thanks{$^\ddag$Research of this author is partially supported by Polish National Science Center (NCN) grant 2015/17/B/ST6/01873.}

\title[Vertex deletion into bipartite permutation graphs]{Vertex deletion into \\bipartite permutation graphs$^*$}

\keywords{permutation graphs, comparability graphs, partially ordered set, graph modification problems}

\begin{document}

\begin{abstract}

A permutation graph can be defined as an intersection graph of segments whose endpoints lie on two parallel lines $\ell_1$ and $\ell_2$, one on each. A bipartite permutation graph is a permutation graph which is bipartite.

In this paper we study the parameterized complexity of the bipartite permutation vertex deletion problem, which asks, for a given $n$-vertex graph, whether we can remove at most $k$ vertices to obtain a bipartite permutation graph. This problem is \NP-complete by the classical result of Lewis and Yannakakis \cite{LewYan78}.

We analyze the structure of the so-called almost bipartite permutation graphs which may contain holes (large induced cycles)
in contrast to bipartite permutation graphs. We exploit the structural properties of the shortest hole in a such graph. We use it to obtain an algorithm for the bipartite permutation vertex deletion problem with running time $\Oh{9^k \cdot n^9}$, and also give a polynomial-time 9-approximation algorithm.

\end{abstract}

\maketitle

\section{Introduction}
\label{sec:introduction}

Many standard computational problems, including maximum clique, maximum independent set, or minimum coloring, which are \NP-hard in general, have poly\-nomial-time exact or approximation algorithms in restricted classes of graphs.
Due to the practical and theoretical applications, some of such graph classes are particularly intensively studied.
Among them~are:
\begin{itemize}
 \item \emph{interval graphs}: intersection graphs of intervals on a real line, 
 \item \emph{unit interval graphs}: intersection graphs of intervals none of
which is contained in another, %
 \item \emph{chordal graphs}: intersection graphs of subtrees of a tree,
 \item \emph{function and permutation graphs}: intersection graphs of continuous and linear functions, respectively, defined on the interval $[0,1]$,
 \item \emph{comparability graphs}: graphs whose edges correspond to the pairs of vertices comparable in some fixed partial order ${<}$ on the vertex set (such an order is called a \emph{transitive orientation} of the graph),
 \item \emph{co-comparability graphs}: the complements of comparability graphs.
\end{itemize}
It is well known that the class of function graphs corresponds to the class of co-comparability graphs \cite{DBLP:journals/dm/GolumbicRU83}, and the class of permutation graphs corresponds to the intersection of comparability and co-comparability graphs \cite{pnueli1971transitive} (see Figure \ref{fig:class-hierarchy} for the hierarchy of inclusions).
All these classes of graphs are \emph{hereditary}, which means 
that they are closed under vertex deletion. 

\begin{figure}[htp!]
\begin{center}
\begin{tikzpicture}[xscale=1,yscale=0.8,>=latex]
\tikzstyle{cla}=[draw,rectangle,fill=white,inner sep=2pt,minimum width=75pt,minimum height=20pt,rounded corners=7pt]
\node[cla] (pf) at (11,0) {perfect};
\node[cla] (intv) at (3.5,1.3) {interval};
\node[cla] (uintv) at (0,1.3) {proper interval};
\node[cla] (bip) at (3.5,-1.3) {bipartite};
\node[cla] (perm) at (3.5,0) {permutation};
\node[cla] (biperm) at (0,-0.65) {\begin{tabular}{c} bipartite \\ permutation
\end{tabular}};

\tikzstyle{cla}=[draw,rectangle,fill=white, inner sep=2pt,minimum width=85pt,minimum height=20pt,rounded corners=7pt]

\node[cla] (chor) at (7.5,1.3) {chordal};
\node[cla] (com) at (7.5,-1.3) {comparability};
\node[cla] (cocom) at (7.5,0) {co-comparability};

\draw[thick,->] (uintv) -- (intv);
\draw[thick,->] (intv) -- (chor);
\draw[thick,->] (intv) -- (cocom);
\draw[thick,->] (chor) -- (pf);
\draw[thick,->] (com) -- (pf);
\draw[thick,->] (cocom) -- (pf);
\draw[thick,->] (biperm) -- (perm);
\draw[thick,->] (biperm) -- (bip);
\draw[thick,->] (perm) -- (com);
\draw[thick,->] (perm) -- (cocom);
\draw[thick,->] (bip) -- (com);
\end{tikzpicture}
\end{center}
\caption{\label{fig:class-hierarchy}
An hierarchy of inclusion of the hereditary graph classes considered in the introduction.
 An arrow from graph class $A$ to graph class $B$ indicates that $A \subset B$.}
\end{figure}
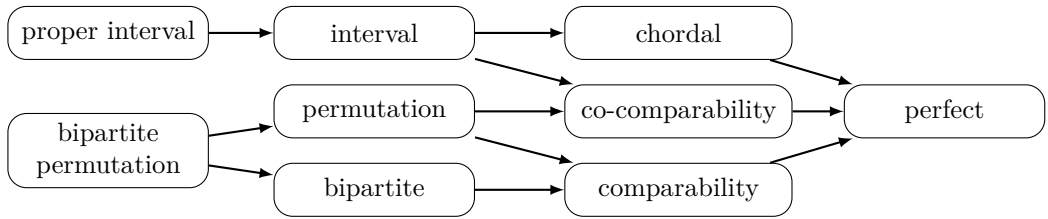

Being hereditary is a very useful property in algorithmic design as every hereditary class of graphs can also be uniquely characterized in terms of minimal forbidden induced subgraphs: a~graph belongs to a class $\calG$ if and only if it does not contain any graph from some family $\calF$ as an induced subgraph. 
For every graph class introduced above, a characterization by forbidden subgraphs is known, see \cite{CRST06} for perfect graphs, \cite{LekBol62} for interval graphs,~\cite{Gal67} for comparability and permutation graphs. 
However, for all of them, the family of forbidden subgraphs is infinite and it may also be quite complex. %
Moreover, every graph $G$ from any class introduced above is \emph{perfect}. %
Gr\"otschel, Lov\'asz, and Schrijver \cite{GLS-book} showed that in the class of perfect graphs the maximum clique, the maximum independent set, and the minimum coloring problems can be solved in polynomial time. 

Polynomial-time algorithms devised for the above-mentioned graph classes can sometimes be adjusted to also work on graphs that are ``close'' to graphs from these classes.
Usually, the ``closeness'' of a graph $G$ to a graph class $\mathcal{G}$
is measured by the number of operations required to transform $G$ into a graph from the class $\mathcal{G}$, where a single operation consists either on removing a vertex from $G$ or on adding or removing an edge from $G$.
Such an approach leads us to the following generic problem.

\medskip
\begin{tabular}{rl}
\textbf{Problem}: & Graph modification problem into a class of graphs $\calG$ \\
\textbf{Input}: & A graph $G$ (typically not from $\calG$) and a number $k$ \\
\textbf{Question}: & Can the graph $G$ be transformed into a graph of the class $\calG$ \\ & by performing at most $k$ modifications of an appropriate kind?
\end{tabular}
\medskip

Depending on the kind of modifications allowed, we obtain four variants of this problem: vertex deletion problem, edge deletion problem, edge completion problem, and edge edition problem (the latter allowing both deletions and additions of edges).
For the class of graphs defined above, all four variants of the modification problem are \NP-hard---see \cite{Man08} for references to \NP-hardness proofs.
In particular, Lewis and Yannakakis \cite{LewYan78} showed that the vertex deletion problem into any non-trivial hereditary class of graphs is \NP-hard.
This is not surprising, as many classical hard problems can be formulated as vertex deletion problems into particular classes of graphs, for example, \textsc{Vertex Cover} as vertex deletion to edgeless graphs, \textsc{Feedback Vertex Set} as vertex deletion to forests, and \textsc{Odd Cycle Transversal} as vertex deletion to bipartite graphs.

Graph modification problems are a popular research direction in the study of the \emph{parameterized complexity} of \NP-complete problems.
In general, for a problem $\Pi$, an input of a parameterized problem consists of an instance $I$ of $\Pi$ and a \emph{parameter} $k \in \mathbb{N}$. Then we say that $\Pi$ is \emph{fixed parameter tractable} (\FPT) if there exists an algorithm deciding whether $I$ is a yes-instance of $\Pi$ in time $f(k)\cdot |I|^{\mathcal{O}(1)}$, where $f$ is some computable function. 
For a graph modification problem, we often choose the parameter $k$ as a number of allowed modifications, so the instance of such a problem is still a pair $(G,k)$.

It turns out that characterizations by forbidden structures are sometimes useful to design \FPT algorithms for graph modification problems.
For example, Cai \cite{Cai96} proposed an \FPT algorithm for modification problems into classes of graphs characterized by a finite family of forbidden induced subgraphs~$\calF$.
His algorithm identifies a forbidden structure in the input graph (which can be done in polynomial time when $\calF$ is finite) and branches over all possible ways of modifying that structure.
Since the families of forbidden structures are infinite for graph classes introduced above, modification algorithms for these classes have to be much more sophisticated.
For several of them modification problems have satisfactory solutions:
\begin{itemize}
\item chordal graphs: all four versions of the modification problem are \FPT \cite{CaoMarx14,Marx06};
\item interval graphs: edge completion and edge deletion are \FPT \cite{VHPTIntCom09,DBLP:conf/soda/Cao16}, vertex deletion is \FPT \cite{CaoMarx14}, edge edition remains open;
\item proper interval graphs: all four versions of the modification problem are \FPT \cite{Cao15}.
\end{itemize}
On the other hand, it is known that the vertex deletion to perfect graphs is 
\W{2}-hard~\cite{HHJKV}.
It is worth mentioning that for a long time, it was unknown whether there are classes of graphs recognizable in polynomial time for which modification problems are hard.
The first such example was given by Lokshtanov~\cite{Lok08}, who proved that the vertex deletion is \W{2}-hard for graphs avoiding all \emph{wheels} (i.e.,\ cycles with an additional vertex adjacent to all other vertices).
It is unknown whether comparability graphs, co-comparability graphs, and permutation graphs have \FPT modification algorithms.
The class of co-comparability graphs,
which constitutes the superclass of interval graphs and an important subclass of perfect graphs, seems to be particularly interesting from the parameterized point of view.

\subsection*{Our focus} Like the class of interval graphs, the class of permutation graphs admits polynomial-time algorithms for rich family problems which are \NP-complete in general. 
Apart from the already mentioned classical hard problems which are polynomial-time solvable for perfect graphs, there also exist polynomial algorithms solving e.g., \textsc{Hamiltonian Cycle}, \textsc{Feedback Vertex Set} or \textsc{Dominating Set} in the class of permutation graphs~\cite{DBLP:conf/fct/BrandstadtK85,DBLP:journals/siamcomp/DeogunS94}.

In light of the above considerations, since all the  modification problems into the class of permutation graphs---and the related classes of comparability and co-comparability graphs---remain open, restricting our attention to the class of \emph{bipartite permutation graphs} appears to be a natural research direction.
Bipartite permutation graphs form an interesting graph class themselves, first investigated by Spinrad, Brandst\"{a}dt, and Stewart~\cite{SBS87}, who characterized them by means of appropriately chosen linear orderings of its bipartition classes. 

One of the most interesting results concerning the bipartite permutation graphs is by Heggernes et al.~\cite{HHLN12}, who showed that the \NP-complete problem of computing the \emph{cutwidth} of a graph (i.e., finding a linear order of the vertices of a graph that minimizes the maximum number of edges intersected by any line inserted between two consecutive vertices) is polynomial for bipartite permutation graphs. 

Our algorithm exploits the absence of some forbidden structures in bipartite permutation graphs. Since these structures cannot, in particular, occur in permutation graphs, we believe that besides being a complete result itself, our research is a step towards understanding the parameterized complexity of modification problems into permutation graphs.

\subsection*{Our results} 
We focus mainly on the modification by vertex deletion. 
\begin{theorem}\label{thm_main_fpt}
There is an $\Oh{9^k \cdot |V(G)|^9}$-time algorithm for instances $(G,k)$ of the vertex deletion into bipartite permutation graphs problem.
\end{theorem}
We prove Theorem~\ref{thm_main_fpt} in Section~\ref{sec_fpt_proof}.
Our algorithm is based on the characterization of bipartite permutation graphs by forbidden subgraphs. 
Using the characterization, at first, we get rid of constant-size forbidden subgraphs by branching, which is a standard technique in modification problems on hereditary graph classes~\cite{HV13,VHPTIntCom09}. 
We call graphs without these forbidden subgraphs \emph{almost bipartite permutation graphs}. 

Our main contribution is in the structural analysis of almost bipartite permutation graphs which may contain holes %
(on more than ten vertices)
in contrast to bipartite permutation graphs. This approach is partially inspired by the ideas of van 't Hof and Villanger \cite{HV13} who used similar tools in their work on proper interval vertex deletion problem. We use the result of
Spinrad, Brandst\"{a}dt, and Stewart~\cite{SBS87}, who showed that the vertices of every connected bipartite permutation graph $G=(U,W,E)$ can be embedded into a strip in such a way that
the vertices from $U$ are on the bottom edge of the strip, the vertices from $W$ are on the top edge of the strip, 
the neighbors $N(u)$ of $u$ occur consecutively on the top edge of the strip for every $u \in U$ (adjacency property), 
the vertices from $N(u)-N(u')$ occur consecutively on the top edge of the strip for every $u,u' \in U$ (enclosure property), 
and the analogous properties are satisfied by the vertices in $W$ (see Figure~\ref{fig:bipartite_permutation_graph}).

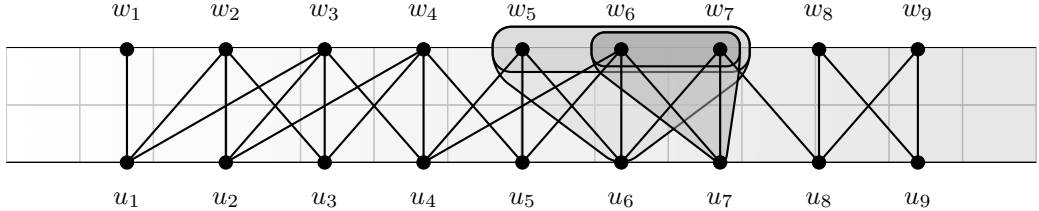
\begin{figure}[htp!]
\centering

\begin{tikzpicture}[xscale=1.3,yscale=1.5]

\begin{axis}[height=2.6cm, width=12cm,
    hide axis,
    view = {0}{90},
    at={(-1.22cm,0)}
    ]
 \addplot3 [
    surf,
    colormap={blackwhite}{gray(0cm)=(1); gray(1cm)=(0.9)},
    shader     = faceted interp,%opacity = 0.7,
    %shader = interp,
    point meta = x,
    samples    = 15,
    samples y  = 3,
    z buffer   = sort,
    domain     = -.5:8.5,
    y domain   = 0:1
    ] (
    {x},
    {y/4},
    {0}
    );
 \addplot3 [color=black,%thick,    
    domain     = -.5:8.5,samples y=0,samples=2*(640/360)*24+1,
    ] (
    {x},
    {0},
    {0} 
    );
    
    \addplot3 [color=black,%thick,    
    domain     = -.5:8.5,samples y=0,samples=2*(640/360)*24+1,
    ] (
    {x},
    {1/4},
    {0}
    );   
    
\end{axis}

\coordinate (w1) at (0,1) {};
\coordinate (w2) at (1,1) {};
\coordinate (w3) at (2,1) {};
\coordinate (w4) at (3,1) {};
\coordinate (w5) at (4,1) {};
\coordinate (w6) at (5,1) {};
\coordinate (w7) at (6,1) {};
\coordinate (w8) at (7,1) {};
\coordinate (w9) at (8,1) {};

\coordinate (lw1) at (0,1.33) {};
\coordinate (lw2) at (1,1.33) {};
\coordinate (lw3) at (2,1.33) {};
\coordinate (lw4) at (3,1.33) {};
\coordinate (lw5) at (4,1.33) {};
\coordinate (lw6) at (5,1.33) {};
\coordinate (lw7) at (6,1.33) {};
\coordinate (lw8) at (7,1.33) {};
\coordinate (lw9) at (8,1.33) {};

\coordinate (u1) at (0,0) {};
\coordinate (u2) at (1,0) {};
\coordinate (u3) at (2,0) {};
\coordinate (u4) at (3,0) {};
\coordinate (u5) at (4,0) {};
\coordinate (u6) at (5,0) {};
\coordinate (u7) at (6,0) {};
\coordinate (u8) at (7,0) {};
\coordinate (u9) at (8,0) {};

\coordinate (lu1) at (0,-0.33) {};
\coordinate (lu2) at (1,-0.33) {};
\coordinate (lu3) at (2,-0.33) {};
\coordinate (lu4) at (3,-0.33) {};
\coordinate (lu5) at (4,-0.33) {};
\coordinate (lu6) at (5,-0.33) {};
\coordinate (lu7) at (6,-0.33) {};
\coordinate (lu8) at (7,-0.33) {};
\coordinate (lu9) at (8,-0.33) {};

\begin{scope}[fill opacity=0.5]
\draw[rounded corners=7, fill=gray!30, thick] (3.7,0.8)--(3.7,1.2) -- (6.3,1.2) -- (6.3,0.8)--(5,-0.05)--cycle;
\draw[rounded corners=7, fill=gray!40, thick] (3.7,0.8)--(3.7,1.2) -- (6.3,1.2) -- (6.3,0.8)--cycle;
\draw[rounded corners=5, fill=gray!60, thick] (4.7,0.85)--(4.7,1.15) -- (6.2,1.15) -- (6.2,0.85)--(6.05,-0.05)--cycle;
\draw[rounded corners=5, fill=gray!70, thick] (4.7,0.85)--(4.7,1.15) -- (6.2,1.15) -- (6.2,0.85)--cycle;
\end{scope}

%\path (u6) edge[thick] (3.77,0.9);
%\path (u6) edge[thick] (6.23,0.9);

\path (u1) edge[thick] (w1);
\path (u1) edge[thick] (w2);
\path (u1) edge[thick] (w3);

\path (u2) edge[thick] (w2);
\path (u2) edge[thick] (w3);
\path (u2) edge[thick] (w4);

\path (u3) edge[thick] (w2);
\path (u3) edge[thick] (w3);
\path (u3) edge[thick] (w4);

\path (u4) edge[thick] (w3);
\path (u4) edge[thick] (w4);
\path (u4) edge[thick] (w5);
\path (u4) edge[thick] (w6);

\path (u5) edge[thick] (w4);
\path (u5) edge[thick] (w5);
\path (u5) edge[thick] (w6);

\path (u6) edge[thick] (w5);
\path (u6) edge[thick] (w6);
\path (u6) edge[thick] (w7);

\path (u7) edge[thick] (w6);
\path (u7) edge[thick] (w7);

\path (u8) edge[thick] (w7);
\path (u8) edge[thick] (w8);
\path (u8) edge[thick] (w9);

\path (u9) edge[thick] (w8);
\path (u9) edge[thick] (w9);

\tikzstyle{every node}=[circle,minimum size=5pt,inner sep=0pt,draw,fill]
\node at (w1) {};
\node at (w2) {};
\node at (w3) {};
\node at (w4) {};
\node at (w5) {};
\node at (w6) {};
\node at (w7) {};
\node at (w8) {};
\node at (w9) {};
\node at (u1) {};
\node at (u2) {};
\node at (u3) {};
\node at (u4) {};
\node at (u5) {};
\node at (u6) {};
\node at (u7) {};
\node at (u8) {};
\node at (u9) {};

\tikzstyle{every node}=[inner sep=2pt]
%\begin{footnotesize}
\node at (lw1) {$w_1$};
\node at (lw2) {$w_2$};
\node at (lw3) {$w_3$};
\node at (lw4) {$w_4$};
\node at (lw5) {$w_5$};
\node at (lw6) {$w_6$};
\node at (lw7) {$w_7$};
\node at (lw8) {$w_8$};
\node at (lw9) {$w_9$};

\node at (lu1) {$u_1$};
\node at (lu2) {$u_2$};
\node at (lu3) {$u_3$};
\node at (lu4) {$u_4$};
\node at (lu5) {$u_5$};
\node at (lu6) {$u_6$};
\node at (lu7) {$u_7$};
\node at (lu8) {$u_8$};
\node at (lu9) {$u_9$};
%\end{footnotesize}

\end{tikzpicture}

\caption{\label{fig:bipartite_permutation_graph}
Embedding of a bipartite permutation graph $(U,W,E)$ into a strip satysfying the adjacency and the enclosure properties.
}
\end{figure}

Our structural result asserts that, depending on the parity of the length of the shortest hole, a connected almost bipartite permutation graph may be naturally 
embedded in either a cylinder, or a M\"obius strip, locally satisfying adjacency and enclosure properties (see~Figure~\ref{fig:cylinder}).

\begin{figure}[htp!]
\pgfplotsset{compat=1.16}

\newcommand{\ptc}[3]{(axis cs:{(1+0.5*(0.35)*cos(#1+#2)},{(1+0.5*(0.35)*sin(#1+#2)},{0.5*(0.35+#3)})}
\newcommand{\pbc}[3]{(axis cs:{(1+0.5*(0.35)*cos(#1+#2)},{(1+0.5*(0.35)*sin(#1+#2)},{-0.5*(0.35+#3)})}

\newcommand{\pms}[3]{(axis cs:{(1+0.5*(0.35+#3)*cos((#1+#2)/2))*cos(#1+#2)},{(1+0.5*(0.35+#3)*cos((#1+#2)/2))*sin(#1+#2)},{(0.5*(0.35+#3)*sin((#1+#2)/2))})}

\centering
\tikzset{
    set arrow inside/.code={\pgfqkeys{/tikz/arrow inside}{#1}},
    set arrow inside={end/.initial=>, opt/.initial=},
    /pgf/decoration/Mark/.style={
        mark/.expanded=at position #1 with
        {
            \noexpand\arrow[\pgfkeysvalueof{/tikz/arrow inside/opt}]{\pgfkeysvalueof{/tikz/arrow inside/end}}
        }
    },
    arrow inside/.style 2 args={
        set arrow inside={#1},
        postaction={
            decorate,decoration={
                markings,Mark/.list={#2}
            }
        }
    },
} %from here: https://tex.stackexchange.com/questions/163689/add-arrows-to-a-smooth-tikz-function
 \begin{tikzpicture}[xscale=0.8,yscale=0.9]
    \begin{axis}[
    hide axis,
    view = {0}{50}
    ]
%    \addplot3 [color=blue,thick,    
%    domain     = 360:720,samples y=0,
%    ] (
%    {(1+0.5*0.5*cos(x/2)))*cos(x)},
%    {(1+0.5*0.5*cos(x/2)))*sin(x)},
%    {0.5*0.5*sin(x/2)}
%    );
    \addplot3 [
    surf,
    colormap={blackwhite}{gray(0cm)=(1); gray(1cm)=(0.75)},
    shader     = faceted interp,%opacity = 0.7,
    %shader = interp,
    point meta = x,
    samples    = 40,
    samples y  = 3,
    z buffer   = sort,
    domain     = 0:360,
    y domain   =-0.35:0.35
    ] (
    {(1+0.5*y*cos(x/2)))*cos(x)},
    {(1+0.5*y*cos(x/2)))*sin(x)},
    {0.5*y*sin(x/2)}
    );
    \
   
    \addplot3 [color=black,%thick,    
    domain     = -173:489.5,samples y=0,samples=2*(640/360)*24+1,
    ] (
    {(1+0.5*0.35*cos(x/2)))*cos(x)},
    {(1+0.5*0.35*cos(x/2)))*sin(x)},
    {0.5*0.35*sin(x/2)}
    ) [arrow inside={end=stealth,opt={black,scale=1.5}}{0.24,0.53,0.9}];

    \draw \pms{240}{0}{0} coordinate (u5);
    \draw \pms{256}{0}{0} coordinate (u6); %u6
    \draw \pms{272}{0}{0} coordinate (u7); %u7
    \draw \pms{288}{0}{0} coordinate (u8);
 
    \draw \pms{600}{0}{0} coordinate (w4);
    \draw \pms{614}{0}{0} coordinate (w5); %w5
    \draw \pms{628}{0}{0} coordinate (w6); %w6
    \draw \pms{642}{0}{0} coordinate (w7); %w7
    \draw \pms{656}{0}{0} coordinate (w8);
 
\begin{scope}[fill opacity=0.5]
\draw[rounded corners=8, fill=gray!60, thick] \pms{614}{-5}{.15}--\pms{628}{0}{.15}--\pms{642}{7}{.15}--\pms{256}{-1}{0.05}--cycle;
\draw[rounded corners=8, fill=gray!80, thick] \pms{628}{-5}{.1}--\pms{642}{5}{.1}--\pms{272}{-.5}{0.07}--cycle;
%\draw[rounded corners=7, fill=gray!40, thick] (3.7,0.8)--(3.7,1.2) -- (6.3,1.2) -- (6.3,0.8)--cycle;
%\draw[rounded corners=5, fill=gray!60, thick] (4.7,0.85)--(4.7,1.15) -- (6.2,1.15) -- (6.2,0.85)--(6.05,-0.05)--cycle;
%\draw[rounded corners=5, fill=gray!70, thick] (4.7,0.85)--(4.7,1.15) -- (6.2,1.15) -- (6.2,0.85)--cycle;
\end{scope}

\path (u5) edge[thick] (w4);
\path (u5) edge[thick] (w5);
\path (u5) edge[thick] (w6);

\path (u6) edge[thick] (w5);
\path (u6) edge[thick] (w6);
\path (u6) edge[thick] (w7);

\path (u7) edge[thick] (w6);
\path (u7) edge[thick] (w7);

\path (u8) edge[thick] (w7);
\path (u8) edge[thick] (w8);

  \tikzstyle{every node}=[circle,minimum size=3pt,inner sep=0pt,draw,fill]

\node at (w4) {};
\node at (w5) {};
\node at (w6) {};
\node at (w7) {};
\node at (w8) {};

\node at (u5) {};
\node at (u6) {};
\node at (u7) {};
\node at (u8) {};

     \end{axis}
 \end{tikzpicture}\qquad
 \begin{tikzpicture}[xscale=0.8,yscale=0.9]
    \begin{axis}[
    hide axis,
    view = {0}{50}
    ]
%    \addplot3 [color=blue,thick,    
%    domain     = 360:720,samples y=0,
%    ] (
%    {(1+0.5*0.5*cos(x/2)))*cos(x)},
%    {(1+0.5*0.5*cos(x/2)))*sin(x)},
%    {0.5*0.5*sin(x/2)}
%    );
    \addplot3 [
    surf,
    colormap={blackwhite}{gray(0cm)=(1); gray(1cm)=(0.75)},
    shader     = faceted interp,%opacity = 0.7,
    %shader = interp,
    point meta = x,
    samples    = 40,
    samples y  = 3,
    z buffer   = sort,
    domain     = 0:360,
    y domain   =-0.35:0.35
    ] (
    {(1+0.5*.35*cos(x)},
    {(1+0.5*.35*sin(x)},
    {0.5*y}
    );

    \addplot3 [color=black,%thick,    
    domain     = 0:360,samples y=0,samples=2*(640/360)*24+1,
    ] (
    {(1+0.5*.35*cos(x)},
    {(1+0.5*.35*sin(x)},
    {0.5*.35}
    ) [arrow inside={end=stealth,opt={black,scale=1.5}}{0.58,0.92}];
    \addplot3 [color=black,%thick,    
    domain     = 180:360,samples y=0,samples=2*(640/360)*24+1,
    ] (
    {(1+0.5*.35*cos(x)},
    {(1+0.5*.35*sin(x)},
    {-0.5*.35} 
    ) [arrow inside={end=stealth,opt={black,scale=1.5}}{0.16,0.84}];
    
    \addplot3 [color=black,%thick,    
    domain     = 44.5:135.5,samples y=0,samples=2*(640/360)*24+1,
    ] (
    {(1+0.5*.35*cos(x)},
    {(1+0.5*.35*sin(x)},
    {-0.5*.35}
    );

    \draw \ptc{240}{0}{0} coordinate (u5);
    \draw \ptc{256}{0}{0} coordinate (u6); %u6
    \draw \ptc{272}{0}{0} coordinate (u7); %u7
    \draw \ptc{288}{0}{0} coordinate (u8);
 
    \draw \pbc{240}{0}{0} coordinate (w4);
    \draw \pbc{254}{0}{0} coordinate (w5); %w5
    \draw \pbc{268}{0}{0} coordinate (w6); %w6
    \draw \pbc{282}{0}{0} coordinate (w7); %w7
    \draw \pbc{296}{0}{0} coordinate (w8);
 
\begin{scope}[fill opacity=0.5]
\draw[rounded corners=8, fill=gray!60, thick] \pbc{254}{-5}{.08}--\pbc{268}{0}{.065}--\pbc{282}{7}{.08}--\ptc{256}{-1}{0.05}--cycle;
\draw[rounded corners=8, fill=gray!80, thick] \pbc{268}{-5}{.05}--\pbc{282}{5}{.05}--\ptc{272}{0}{0.05}--cycle;
%\draw[rounded corners=7, fill=gray!40, thick] (3.7,0.8)--(3.7,1.2) -- (6.3,1.2) -- (6.3,0.8)--cycle;
%\draw[rounded corners=5, fill=gray!60, thick] (4.7,0.85)--(4.7,1.15) -- (6.2,1.15) -- (6.2,0.85)--(6.05,-0.05)--cycle;
%\draw[rounded corners=5, fill=gray!70, thick] (4.7,0.85)--(4.7,1.15) -- (6.2,1.15) -- (6.2,0.85)--cycle;
\end{scope}

\path (u5) edge[thick] (w4);
\path (u5) edge[thick] (w5);
\path (u5) edge[thick] (w6);

\path (u6) edge[thick] (w5);
\path (u6) edge[thick] (w6);
\path (u6) edge[thick] (w7);

\path (u7) edge[thick] (w6);
\path (u7) edge[thick] (w7);

\path (u8) edge[thick] (w7);
\path (u8) edge[thick] (w8);

  \tikzstyle{every node}=[circle,minimum size=3pt,inner sep=0pt,draw,fill]

\node at (w4) {};
\node at (w5) {};
\node at (w6) {};
\node at (w7) {};
\node at (w8) {};

\node at (u5) {};
\node at (u6) {};
\node at (u7) {};
\node at (u8) {};

     \end{axis}
 \end{tikzpicture}
 
\caption{\label{fig:cylinder} Embedding of an almost bipartite permutation graph in a cylinder or M\"obius~strip.}% locally satisfies adjacency and enclosure properties.}

\end{figure}
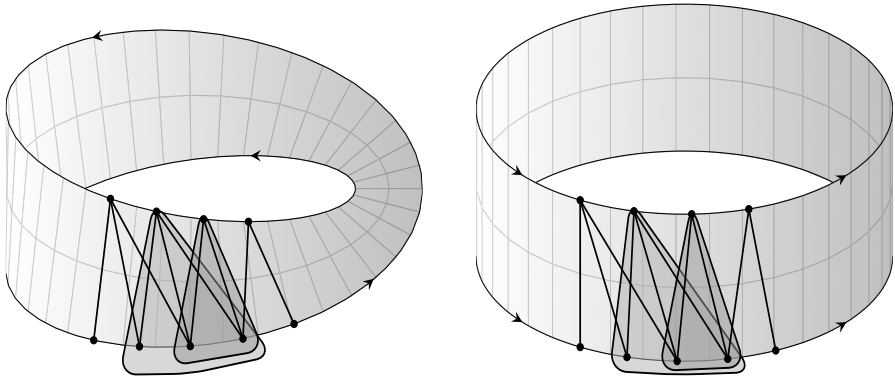

Once we obtain such structure, we show that every minimal vertex cut that destroys all holes lies nearby a few consecutive vertices from the shortest hole.
This allows us to check all the possibilities where we can find a minimum cut.
Finally, we use a polynomial algorithm for finding maximum flow (and thus a minimum cut).

The approach used to prove Theorem \ref{thm_main_fpt} can be slightly modified to obtain a $9$-approxima\-tion algorithm for the bipartite permutation vertex deletion problem. We show the following.%
\begin{theorem}\label{thm_main_apx}
There exists a polynomial-time $9$-approximation algorithm for vertex deletion into bipartite permutation graphs problem.
\end{theorem}

\section{Preliminaries}

Unless stated otherwise, all graphs considered in this work are simple, i.e., undirected, with no loops and parallel edges.
Let $G=(V,E)$ be a graph.
For a subset $S \subseteq V$, the subgraph of $G$ induced by $S$ is the graph $G[S] = (S, \{uv \mid uv \in E, u,v\in S\})$.
The \emph{neighborhood} of a vertex $u \in V$ is the set $N(u) = \{v \in V \mid uv \in E\}$.
Similarly, we write $N(U)=\bigcup_{u\in U} N(u)\setminus U$ for a set $U\subseteq V$. 
Let $u,v\in V$. 
We say that $u$ and $v$ \emph{are at distance~$k$ (in $G$)} if $k$ is the length of a shortest path between $u$ and $v$ in $G$.
We denote a complete graph and a cycle on $n$ vertices by $K_n$ and $C_n$, respectively.
By \emph{hole} we mean an induced cycle on at least five vertices.
We say that a hole is \emph{even} (or \emph{odd}) if it contains even (odd) number of vertices, respectively.

For a graph $G=(V,E)$, a pair $(V,{<})$ is a \emph{transitive orientation} of $G$ if ${<}$ is a transitive and irreflexive relation on $V$ that satisfies
either $u < v$ or $v < u$ iff $uv \in E$ for every $u,v \in V$.

A \emph{partially ordered set} (shortly \emph{partial order} or \emph{poset}) is a pair $P = (X,{\leq_P})$ that consists of a set $X$ and a reflexive, transitive, and antisymmetric relation ${\leq_P}$ on $X$.
For a poset $(X,{\leq_P})$, let the \emph{strict partial order} $<_P$ be a binary relation defined on $X$ such that $x <_P y$ if and only if $x \leq_P y$ and $x \neq y$. Equivalently, $(X,{<_P})$ is a strict partial order if $<_P$ is irreflexive and transitive.
Two elements $x,y \in X$ are \emph{comparable} in $P$ if $x \leq_P y$ or $y \leq_P x$; otherwise, $x,y$ are \emph{incomparable} in $P$.
A \emph{linear order} $L=(X,{\leq_L})$ is a partial order in which for every $x,y \in X$ we have $x \leq_L y$ or $y \leq_L x$.
A \emph{strict linear order} $(X,<_L)$ is a binary relation defined in a way that $x <_L y$ if and only if $x \leq_L y$ and $x \neq y$. 

Let $P=(X,{\leq_P})$ be a poset. A linear order $L=(X,{\leq_L})$ is called a \emph{linear extension} of $P$ if ${\leq_{P}} \subseteq {\leq_L}$.
Given a family of posets $\mathcal{P} = \{P_i = (X,{\leq_{P_i}}): i \in I\}$, we say that $P$ is the \emph{intersection of $\mathcal{P}$} if
for every $x,y \in X$ we have $x \leq_P y$ if and only if $x \leq_{P_i} y$ for every $i \in I$.  
The \emph{dimension} of a poset $P$ is the minimal number of linear extensions of $P$ that intersect to $P$.
We say that $P$ is \emph{two-dimensional} if it is the intersection of two linear extensions of~$P$.

A \emph{comparability graph} (\emph{incomparability graph}) of a poset $P=(X,{\leq_P})$ has $X$ as the set of its vertices and the set including every two vertices comparable (incomparable, respectively) in
$P$ as the set of its edges.
Note the following: if $(X,{\leq_P})$ is a poset, then $(X, {<_P})$ is a transitive orientation of the comparability graph of $P$.
A graph $G=(V,E)$ is a \emph{comparability graph} (\emph{co-comparability graph}) if $G$ is a comparability (incomparability, respectively) graph of some poset defined on $V$.
So, $G$ is a comparability graph if and only if $G$ admits a transitive orientation.
A graph $G$ is a \emph{permutation graph} if and only if $G$ and the complement of $G$ are comparability graphs \cite{pnueli1971transitive} (or equivalently, $G$ and the complement of $G$ admit transitive orientations).
Baker, Fishburn, and Roberts  \cite{baker1972partial} proved that $G$ is a permutation graph if and only if $G$ is the incomparability graph of a two-dimensional poset.

We say that two sets $X$ and $Y$ are \emph{comparable} if $X$ and $Y$ are comparable with respect to $\subseteq$-relation (that is, $X \subseteq Y$ or $Y \subseteq X$ holds).
We use the convenient notation $[m]:=\{0,1,\dots,m\},$ for every $m \in \mathbb{N}$.
For every $i,j \in \mathbb{Z}$ such that $i \leq j$ by $[i,j]$ we mean the set $\{i,i+1,\ldots,j\}$.

\section{The structure of (almost) bipartite permutation graphs}\label{sec_structure}
\label{sec:locally_bipartite_permutation_graphs}

The characterization of bipartite permutation graphs presented below was proposed by
Spinrad, Brandst\"{a}dt, and Stewart~\cite{SBS87}.

Suppose $G=(U,W,E)$ is a connected bipartite graph.
A linear order $(W,{<_W})$ satisfies \emph{adjacency property} 
if for each vertex $u \in U$ the set $N(u)$ consists of vertices that are consecutive in $(W,{<_W})$.
A linear order $(W,{<_W})$ satisfies \emph{enclosure property}
if for every pair of vertices $u,u' \in U$ such that $N(u)$ is a subset of $N(u')$, vertices in $N(u') - N(u)$ occur consecutively in $(W,{<_W})$.
A \emph{strong ordering} of the vertices of $U \cup W$ consists of 
linear orders $(U,{<_U})$ and $(W,{<_W})$ such that for every
$(u,w'), (u',w)$ in $E$, where $u,u'$ are in $U$ and $w,w'$ are in $W$,
$u <_U u'$ and $w <_W w'$ imply $(u,w) \in E$ and $(u',w') \in E$.
Note that, whenever $(U,{<_U})$ and $(W,{<_W})$ form a strong ordering of $U \cup W$, 
then $(U,{<_U})$ and $(W,{<_W})$ satisfy the adjacency and enclosure properties.
\begin{theorem}[Spinrad, Brandst\"{a}dt, Stewart~\cite{SBS87}]\label{thm:bip_char} The following three statements are equivalent for a connected bipartite graph $G=(U,W,E)$:
\begin{enumerate}
\item \label{thm:bip_char_1} $(U,W,E)$ is a bipartite permutation graph.
\item \label{thm:bip_char_2} There exists a strong ordering of $U \cup W$.
\item \label{thm:bip_char_3} There exists a linear order $(W,{<_W})$ of $W$ satisfying adjacency and enclosure properties.
\end{enumerate}
\end{theorem}
An example of a bipartite permutation graph $G=(U,W,E)$ with linear order $w_1 <_W w_2 <_W\ldots <_W w_8 <_W w_9$ of the vertices of $W$ which satisfies the adjacency and the enclosure properties is shown in Figure \ref{fig:bipartite_permutation_graph}.

Another characterization of bipartite permutation graphs can be obtained by listing all minimal forbidden induced subgraphs for this class of graphs.
Such a list can be compiled by taking all odd cycles of length $\geq 3$ (forbidden structures for bipartite graphs) and all bipartite graphs from the list of forbidden structures for permutation graphs obtained by Gallai~\cite{Gal67}. 
The whole list is shown in Figure~\ref{fig:bp_forbidden_structures}.
\begin{figure}[htp!]
\centering
\begin{tikzpicture}[xscale=0.8,yscale=1]
\coordinate (x1) at (0,0) {};
\coordinate (x2) at (1,0) {};
\coordinate (x3) at (2,0) {};
\coordinate (x4) at (3,0) {};
\coordinate (y1) at (0.5,1) {};
\coordinate (y2) at (1.5,1) {};
\coordinate (y3) at (2.25,1) {};
\coordinate (l) at (1.5,-0.5) {};

\tikzstyle{every node}=[circle,minimum size=5pt,inner sep=0pt,draw,fill]
\node at (x1) {};
\node at (x2) {};
\node at (x3) {};
\node at (x4) {};

\node at (y1) {};
\node at (y2) {};
\node at (y3) {};

\tikzstyle{every node}=[inner sep=1pt]

\path (y1) edge[thick] (x1);
\path (y1) edge[thick] (x2);

\path (y2) edge[thick] (x2);
\path (y2) edge[thick] (x3);

\path (y3) edge[thick] (x2);
\path (y3) edge[thick] (x4);

\begin{footnotesize}
\tikzstyle{every node}=[inner sep=2pt]
\node at (l) {$T_2$};
\end{footnotesize}
\end{tikzpicture}
\hspace{0.65cm}
\begin{tikzpicture}[xscale=0.8,yscale=1]
\coordinate (x1) at (0,0) {};
\coordinate (x2) at (1,0) {};
\coordinate (y1) at (0,1) {};
\coordinate (y2) at (1,1) {};
\coordinate (ay1) at (-1,0) {};
\coordinate (ax1) at (-1,1) {};
\coordinate (ax2) at (2,1) {};
\coordinate (l) at (0.5,-0.5) {};

\tikzstyle{every node}=[circle,minimum size=5pt,inner sep=0pt,draw,fill]
\node at (x1) {};
\node at (x2) {};
\node at (y1) {};
\node at (y2) {};

\node at (ax1) {};
\node at (ax2) {};
\node at (ay1) {};

\tikzstyle{every node}=[inner sep=1pt]

\path (y1) edge[thick] (x1);
\path (y1) edge[thick] (x2);
\path (y2) edge[thick] (x1);
\path (y2) edge[thick] (x2);

\path (x1) edge[thick] (ax1);
\path (x2) edge[thick] (ax2);
\path (y1) edge[thick] (ay1);

\begin{footnotesize}
\tikzstyle{every node}=[inner sep=2pt]
\node at (l) {$X_2$};
\end{footnotesize}
\end{tikzpicture}
\hspace{0.65cm}
\begin{tikzpicture}[xscale=0.8,yscale=1]
\coordinate (x1) at (0,0) {};
\coordinate (x2) at (1,0) {};
\coordinate (x3) at (2.5,0) {};
\coordinate (y1) at (0,1) {};
\coordinate (y2) at (1,1) {};
\coordinate (y3) at (2.5,1) {};
\coordinate (ax2) at (1.75,1) {};
\coordinate (l) at (1.25,-0.5) {};

\tikzstyle{every node}=[circle,minimum size=5pt,inner sep=0pt,draw,fill]
\node at (x1) {};
\node at (x2) {};
\node at (x3) {};
\node at (y1) {};
\node at (y2) {};
\node at (y3) {};
\node at (ax2) {};

\tikzstyle{every node}=[inner sep=1pt]

\path (x1) edge[thick] (y1);
\path (x1) edge[thick] (y2);
\path (x2) edge[thick] (y1);
\path (x2) edge[thick] (y2);
\path (x2) edge[thick] (y3);
\path (x3) edge[thick] (y2);
\path (x3) edge[thick] (y3);

\path (x2) edge[thick] (ax2);

\begin{footnotesize}
\tikzstyle{every node}=[inner sep=2pt]
\node at (l) {$X_3$};
\end{footnotesize}
\end{tikzpicture}
\hspace{0.65cm}
\begin{tikzpicture}[xscale=0.8,yscale=1]
\coordinate (x1) at (-0.25,0) {};
\coordinate (x2) at (1,0) {};
\coordinate (x3) at (2.25,0) {};
\coordinate (y1) at (-0.25,1) {};
\coordinate (y2) at (1,1) {};
\coordinate (y3) at (2.25,1) {};
\coordinate (l) at (1,-0.5) {};

\tikzstyle{every node}=[circle,minimum size=5pt,inner sep=0pt,draw,fill]
\node at (x1) {};
\node at (x2) {};
\node at (x3) {};
\node at (y1) {};
\node at (y2) {};
\node at (y3) {};

\tikzstyle{every node}=[inner sep=1pt]

\path (x1) edge[thick] (y3);
\path (x1) edge[thick] (y1);
\path (x2) edge[thick] (y1);
\path (x2) edge[thick] (y2);
\path (x3) edge[thick] (y2);
\path (x3) edge[thick] (y3);

\begin{footnotesize}
\tikzstyle{every node}=[inner sep=2pt]
\node at (l) {$\text{$C_{2k}$ for $k \geq 3$}$};
\end{footnotesize}
\end{tikzpicture}

\vspace{0.2cm}
\begin{tikzpicture}[xscale=0.8,yscale=1]
\coordinate (x1) at (-0.5,0) {};
\coordinate (x2) at (1,0) {};
\coordinate (y1) at (0.25,1) {};
\coordinate (l) at (0.25,-0.5) {};

\tikzstyle{every node}=[circle,minimum size=5pt,inner sep=0pt,draw,fill]
\node at (x1) {};
\node at (x2) {};
\node at (y1) {};

\tikzstyle{every node}=[inner sep=1pt]

\path (x1) edge[thick] (x2);
\path (x1) edge[thick] (y1);
\path (x2) edge[thick] (y1);
\begin{footnotesize}
\tikzstyle{every node}=[inner sep=2pt]
\node at (l) {$K_3$};
\end{footnotesize}
\end{tikzpicture}
\hspace{1cm}
\begin{tikzpicture}[xscale=0.8,yscale=1]
\coordinate (x1) at (-0.25,0) {};
\coordinate (x2) at (1,0) {};
\coordinate (z1) at (2.25,0.5) {};
\coordinate (y1) at (-0.25,1) {};
\coordinate (y2) at (1,1) {};
\coordinate (l) at (1,-0.5) {};

\tikzstyle{every node}=[circle,minimum size=5pt,inner sep=0pt,draw,fill]
\node at (x1) {};
\node at (x2) {};
\node at (y1) {};
\node at (y2) {};
\node at (z1) {};

\tikzstyle{every node}=[inner sep=1pt]

\path (x1) edge[thick] (x2);
\path (x2) edge[thick] (z1);
\path (z1) edge[thick] (y2);
\path (y2) edge[thick] (y1);
\path (y1) edge[thick] (x1);

\begin{footnotesize}
\tikzstyle{every node}=[inner sep=2pt]
\node at (l) {$\text{$C_{2k+1}$ for $k \geq 2$}$};
\end{footnotesize}
\end{tikzpicture}
\caption{\label{fig:bp_forbidden_structures} 
Forbidden structures for bipartite permutation graphs.}
\end{figure}
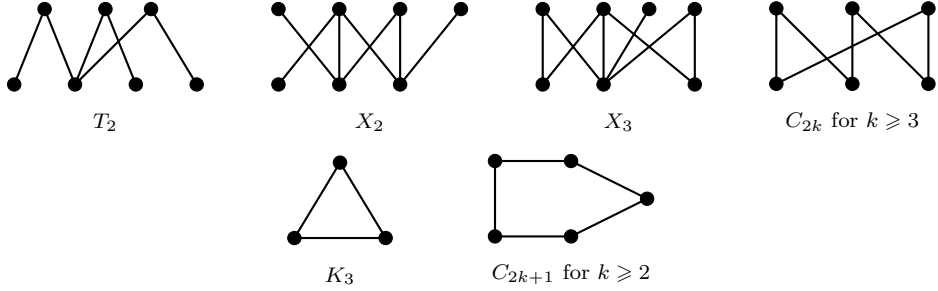

\subsection{Almost bipartite permutation graphs}
\label{sec:almost_bipartite_permutation}
The goal of this section is to characterize graphs which do not contain small 
forbidden subgraphs for the class of bipartite permutation graphs.
Following terminology of van 't Hof and Villanger \cite{HV13} we call such graphs almost bipartite permutation graphs. 
\begin{definition}
A graph $G=(V,E)$ is \emph{almost bipartite permutation graph} if $G$ does not contain $T_2$, $X_2$, $X_3$, $K_3$, $C_k$ for $k \in [5,9]$ as induced subgraphs. %
\end{definition}

Suppose $G=(V,E)$ is a connected almost bipartite permutation graph.
\begin{proposition} \label{prop:C_dominating_set}
Every hole in $G$ is a dominating set.
\end{proposition}
\begin{proof}
Let $C=\{c_0,c_1,\dots,c_{m-1}\}$ be a hole in $G$. 
Hence, $m \geq 10.$
Suppose, for contradiction, that there exists a vertex in the set $V \setminus (C\cup N(C))$. 
As $G$ is connected, there must exist $v \in V$ at distance two from $C$.
Let $w\in N(v)\cap N(C)$ and let $c_j$ be a neighbor of $w$ in $C$.
We now look at the neighborhood of $w$. 
As $G$ contains no triangle, $wc_{j-1}$ and $wc_{j+1}$ are non-edges.
Moreover, as $G$ contains no copy of $T_2$, vertex $w$ is adjacent to at least one of $c_{j-2}$ and $c_{j+2}$, say $c_{j-2}$.
Thus, $w$ is nonadjacent to $c_{j-3}$.
Therefore, the set $\{c_{j-3},c_{j-2},c_{j-1},c_{j},c_{j+1},w,v\}$ induces a copy of $X_2$ in $G$, which leads to a contradiction.
\end{proof}

Let $C$ be a shortest hole in $G$, $m$ be the size of $C$, and $c_0,c_1,\ldots,c_{m-1}$ be the consecutive vertices of $C$, $m \geq 10$.
In the remaining part of the paper we use the following notation with respect to $C$.
For any integral number $i$ by $c_i$ we denote the unique vertex $c_{i \bmod m}$ from the cycle~$C$.
For any two different vertices $c_i,c_j$ in $C$, by \emph{the set of all vertices between $c_i$ and $c_j$ from $C$} we mean the set
$\{c_i,c_{i+1},\ldots,c_{i+k}\}$, where $k$ is the smallest natural number such that $c_{i+k}=c_j$.
Note that this notion is not symmetric, i.e., the set of all vertices between $c_j$ and $c_i$ from $C$ contains $c_i,c_j$ and all the vertices from $C$
that are not between $c_i$ and $c_j$.
\begin{proposition}
\label{prop:C_neighbors}
For every vertex $v \in V$ either:
\begin{enumerate}[(1)]
 \item \label{item:one-neigh} $N(v) \cap C = \{c_i\}$ for some $i \in [m-1]$, or
 \item \label{item:two-neigh} $N(v) \cap C = \{c_{i},c_{i+2}\}$ for some $i \in [m-1]$.
\end{enumerate}
\end{proposition}
\begin{proof}
Since $C$ is an induced cycle, \eqref{item:two-neigh} clearly holds for the vertices from $C$, so let $v$ be a vertex in $V\setminus C$.
As $C$ is a dominating set, by Proposition~\ref{prop:C_dominating_set}, vertex $v$ has at least one neighbor in $C$.
If $v$ has exactly one neighbor in $C$, then \eqref{item:one-neigh} holds and we are done.
So assume that it has more than one neighbor. 
We now distinguish two cases.
First, suppose that there exist two vertices $c_j, c_\ell \in N(v)\cap C$ at distance at least three in $C$ such that  
$v$ has no neighbor in the set of vertices between $c_j$ and $c_l$, except $c_j$ and $c_l$.
Then, $\{c_j,c_{j+1},\dots,c_\ell,v\}$ induces a cycle $C'$ on at least five vertices in $G$.
As $c_j$ and $c_\ell$ are at distance at least three in $C$, $C'$ is shorter than $C$.
In particular, $C'$ contradicts either $G$ containing no copy of $C_\ell$, for $\ell \in \{5,\ldots,9\}$, or $C$ being a shortest hole in $G$.
Therefore, this case never occurs. 

Hence, $v$ has either (i) exactly two neighbors in $C$ and those are at distance two as there is no triangle in $G$, so \eqref{item:two-neigh} holds, or (ii) $C$ has an even number of vertices and $v$ is adjacent to every second vertex of $C$. 
It remains to show that the latter never occurs. 
Indeed, if it does, then without loss of generality $c_0\in N(v)$. 
But observe that since $C$ has at least ten vertices, the set $\{c_0,c_1,c_2,c_3,c_4,c_6,v\}$ induces a copy of $X_3$. This concludes the proof. \qedhere
\end{proof}
Given Proposition~\ref{prop:C_neighbors}, for every $i\in[m-1]$ we can set
$A_i = \big{\{}v \in V: N(v) \cap C = \{c_{i-1},c_{i+1}\}\!\big{\}}$ and $B_i = \big{\{}v \in V: N(v) \cap C = \{c_i\}\big{\}}$.
 Note that $A_0,B_0,\ldots,A_{m-1},B_{m-1}$ is a partition of $V$ and $c_i\in A_i$ .
Following our notation, for any integer $i$ by $A_i$ and $B_i$ we denote the sets $A_{i\bmod m}$ and $B_{i \bmod m}$, respectively. 
Furthermore, for every $i \leq j$ we set:
\begin{align*}
A_G[i,j] &= \left \{ 
 \begin{array}{ll}
    A_i \cup B_{i+1} \cup A_{i+2}\cup B_{i+3} \cup \ldots \cup A_{j-1} \cup B_{j} & \text{if $j-i$ is odd}, \\
    A_i \cup B_{i+1} \cup A_{i+2}\cup B_{i+3} \cup \ldots \cup B_{j-1} \cup A_{j} & \text{if $j-i$ is even}, 
 \end{array}
 \right.
 \\
 B_G[i,j] &= \left \{ 
 \begin{array}{ll}
    B_i \cup A_{i+1} \cup B_{i+2}\cup A_{i+3}\cup \ldots \cup B_{j-1} \cup A_{j} & \text{if $j-i$ is odd}, \\
    B_i \cup A_{i+1} \cup B_{i+2}\cup A_{i+3}\cup \ldots \cup A_{j-1} \cup B_{j} & \text{if $j-i$ is even}, 
 \end{array}
 \right.\\
\text{and \vspace{4pt}~ } V_G[i,j] &= A_G[i,j] \cup B_G[i,j].
\end{align*}

We write just $A[i,j]$, $B[i,j]$, and $V[i,j]$, respectively, instead of $A_G[i,j]$, $B_G[i,j]$, and $V_G[i,j]$, when there is no confusion.

We now characterize the neighborhoods of the vertices in sets $A_i$ and $B_i$, see also Figure~\ref{fig:AiBineigbours}.

\begin{figure}[htp!]
\centering
\begin{tikzpicture}[xscale=2,yscale=1.5]
\coordinate (c') at (-0.5,0.5) {};
\coordinate (c0) at (0,0) {};
\coordinate (c1) at (1,1) {};
\coordinate (c2) at (2,0) {};
\coordinate (c3) at (3,1) {};
\coordinate (c4) at (4,0) {};
\coordinate (u) at (2.25,0) {};
\coordinate (w) at (2,1) {};
\coordinate (c'') at (4.5,0.5) {};

\coordinate (lc0) at (-0.25,-0.05) {};
\coordinate (lc1) at (1.25,1.0) {};
\coordinate (lc2) at (1.85,-0.05) {};
\coordinate (lc3) at (3.25,1.0) {};
\coordinate (lc4) at (4.25,-0.05) {};
\coordinate (lu) at (2.38,-0.05) {};
\coordinate (lw) at (2.15,1.08) {};

\coordinate (lB0) at (0.0,-0.4) {};
\coordinate (lA0) at (0.0,1.4) {};
\coordinate (lB1) at (1.0,-0.4) {};
\coordinate (lA1) at (1.0,1.4) {};
\coordinate (lB2) at (2.0,-0.4) {};
\coordinate (lA2) at (2.0,1.4) {};
\coordinate (lB3) at (3.0,-0.4) {};
\coordinate (lA3) at (3.0,1.4) {};
\coordinate (lB4) at (4.0,-0.4) {};
\coordinate (lA4) at (4.0,1.4) {};

\coordinate (lA00) at (-0.3,0.2) {};
\coordinate (lA01) at (0,0.2) {};
\coordinate (lA02) at (0.3,0.2) {};

\coordinate (lB00) at (-0.3,0.8) {};
\coordinate (lB01) at (0,0.8) {};
\coordinate (lB02) at (0.3,0.8) {};

\coordinate (lA10) at (0.7,0.2) {};
\coordinate (lA11) at (1,0.2) {};
\coordinate (lA12) at (1.3,0.2) {};

\coordinate (lB10) at (0.7,0.8) {};
\coordinate (lB11) at (1,0.8) {};
\coordinate (lB12) at (1.3,0.8) {};

\coordinate (lA20) at (1.7,0.2) {};
\coordinate (lA21) at (2,0.2) {};
\coordinate (lA22) at (2.3,0.2) {};

\coordinate (lB20) at (1.7,0.8) {};
\coordinate (lB21) at (2,0.8) {};
\coordinate (lB22) at (2.3,0.8) {};

\coordinate (lA30) at (2.7,0.2) {};
\coordinate (lA31) at (3,0.2) {};
\coordinate (lA32) at (3.3,0.2) {};

\coordinate (lB30) at (2.7,0.8) {};
\coordinate (lB31) at (3,0.8) {};
\coordinate (lB32) at (3.3,0.8) {};

\coordinate (lA40) at (3.7,0.2) {};
\coordinate (lA41) at (4,0.2) {};
\coordinate (lA42) at (4.3,0.2) {};

\coordinate (lB40) at (3.7,0.8) {};
\coordinate (lB41) at (4,0.8) {};
\coordinate (lB42) at (4.3,0.8) {};

\path (lA00) edge[thick] (lB00);
\path (lA00) edge[thick] (lB01);
\path (lA00) edge[thick] (lB02);
\path (lA01) edge[thick] (lB00);
\path (lA01) edge[thick] (lB01);
\path (lA01) edge[thick] (lB02);
\path (lA02) edge[thick] (lB00);
\path (lA02) edge[thick] (lB01);
\path (lA02) edge[thick] (lB02);

\path (lA10) edge[thick] (lB10);
\path (lA10) edge[thick] (lB11);
\path (lA10) edge[thick] (lB12);
\path (lA11) edge[thick] (lB10);
\path (lA11) edge[thick] (lB11);
\path (lA11) edge[thick] (lB12);
\path (lA12) edge[thick] (lB10);
\path (lA12) edge[thick] (lB11);
\path (lA12) edge[thick] (lB12);

\path (lA20) edge[thick] (lB20);
\path (lA20) edge[thick] (lB21);
\path (lA20) edge[thick] (lB22);
\path (lA21) edge[thick] (lB20);
\path (lA21) edge[thick] (lB21);
\path (lA21) edge[thick] (lB22);
\path (lA22) edge[thick] (lB20);
\path (lA22) edge[thick] (lB21);
\path (lA22) edge[thick] (lB22);

\path (lA30) edge[thick] (lB30);
\path (lA30) edge[thick] (lB31);
\path (lA30) edge[thick] (lB32);
\path (lA31) edge[thick] (lB30);
\path (lA31) edge[thick] (lB31);
\path (lA31) edge[thick] (lB32);
\path (lA32) edge[thick] (lB30);
\path (lA32) edge[thick] (lB31);
\path (lA32) edge[thick] (lB32);

\path (lA40) edge[thick] (lB40);
\path (lA40) edge[thick] (lB41);
\path (lA40) edge[thick] (lB42);
\path (lA41) edge[thick] (lB40);
\path (lA41) edge[thick] (lB41);
\path (lA41) edge[thick] (lB42);
\path (lA42) edge[thick] (lB40);
\path (lA42) edge[thick] (lB41);
\path (lA42) edge[thick] (lB42);

\path (c') edge[very thick] (c0);
\path (c0) edge[very thick] (c1);
\path (c1) edge[very thick] (c2);
\path (c2) edge[very thick] (c3);
\path (c3) edge[very thick] (c4);
\path (c4) edge[very thick] (c'');

\draw[thick] (-0.45,-0.2)--(-0.45,0.2) -- (0.45,0.2) -- (0.45,-0.2)--cycle;
\draw[thick] (-0.45,0.8)--(-0.45,1.2) -- (0.45,1.2) -- (0.45,0.8)--cycle;

\draw[thick] (0.55,-0.2)--(0.55,0.2) -- (1.45,0.2) -- (1.45,-0.2)--cycle;
\draw[thick] (0.55,0.8)--(0.55,1.2) -- (1.45,1.2) -- (1.45,0.8)--cycle;

\draw[thick] (1.55,-0.2)--(1.55,0.2) -- (2.45,0.2) -- (2.45,-0.2)--cycle;
\draw[thick] (1.55,0.8)--(1.55,1.2) -- (2.45,1.2) -- (2.45,0.8)--cycle;

\draw[thick] (2.55,-0.2)--(2.55,0.2) -- (3.45,0.2) -- (3.45,-0.2)--cycle;
\draw[thick] (2.55,0.8)--(2.55,1.2) -- (3.45,1.2) -- (3.45,0.8)--cycle;

\draw[thick] (3.55,-0.2)--(3.55,0.2) -- (4.45,0.2) -- (4.45,-0.2)--cycle;
\draw[thick] (3.55,0.8)--(3.55,1.2) -- (4.45,1.2) -- (4.45,0.8)--cycle;

\begin{scope}[fill opacity=0.5]
\draw[rounded corners=7, fill=gray!20, thick] (0.25, 0.85)--(0.25,1.15) -- (3.75,1.15) -- (3.75,0.85)--(2.25,0)--cycle;
\draw[rounded corners=7, fill=gray!50, thick] (0.25, 0.85)--(0.25,1.15) -- (3.75,1.15) -- (3.75,0.85)--cycle;

\draw[rounded corners=7, fill=gray!20, thick] (0.25, 0.15)--(0.25,-0.15) -- (3.75,-0.15) -- (3.75,0.15)--(2,1.05)--cycle;
\draw[rounded corners=7, fill=gray!50, thick] (0.25, 0.15)--(0.25,-0.15) -- (3.75,-0.15) -- (3.75,0.15)--cycle;
\end{scope}

\draw[rounded corners=7, thick] (0.25, 0.85)--(0.25,1.15) -- (3.75,1.15) -- (3.75,0.85)--(2.25,0)--cycle;
\draw[rounded corners=7, thick] (0.25, 0.85)--(0.25,1.15) -- (3.75,1.15) -- (3.75,0.85)--cycle;

\tikzstyle{every node}=[circle,minimum size=5pt,inner sep=0pt,draw,fill]
\node at (c0) {};
\node at (c1) {};
\node at (c2) {};
\node at (c3) {};
\node at (c4) {};
\node at (u) {};
\node at (w) {};

\tikzstyle{every node}=[inner sep=2pt]
\begin{footnotesize}
\node at (lu) {$u$};
\node at (lw) {$w$};
\node at (lc0) {$c_{i-2}$};
\node at (lc1) {$c_{i-1}$};
\node at (lc2) {$c_{i}$};
\node at (lc3) {$c_{i+1}$};
\node at (lc4) {$c_{i+2}$};
\end{footnotesize}

\node at (lA0) {$B_{i-2}$};
\node at (lB0) {$A_{i-2}$};
\node at (lA1) {$A_{i-1}$};
\node at (lB1) {$B_{i-1}$};
\node at (lA2) {$B_{i}$};
\node at (lB2) {$A_{i}$};
\node at (lA3) {$A_{i+1}$};
\node at (lB3) {$B_{i+1}$};
\node at (lA4) {$B_{i+2}$};
\node at (lB4) {$A_{i+2}$};

\end{tikzpicture}

\caption{A possible neighborhood of $u$ in $A_i$ and $w$ in $B_i$.}
\label{fig:AiBineigbours}
\end{figure}
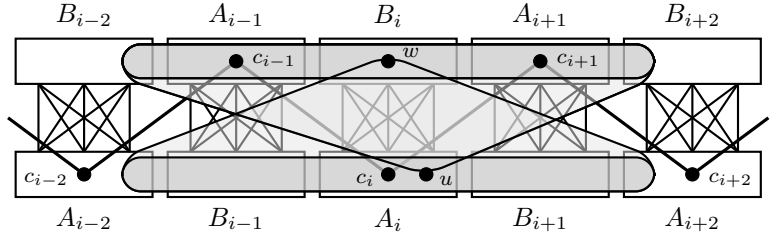
\begin{proposition}\label{prop:neighbourhoodd_containment}%
Let $i \in [m-1]$. Then:
\begin{enumerate}[(1)]
 \item \label{item:neighbourhood_containment_is} $A_i$ and $B_i$ are independent sets.
 \item \label{item:neighbourhood_containment_AiBi} For every $u \in A_i$ and every $w \in B_i$ we have $uw \in E$.
 \item \label{item:neighbourhood_containment_Ai} For every $u \in A_i$ we have $B_i \subseteq N(u) \subseteq B[i-2,i+2]$.
 \item \label{item:neighbourhood_containment_Bi} For every $w \in B_i$ we have $A_i \subseteq N(w) \subseteq A[i-2,i+2]$.
\end{enumerate}
\end{proposition}

\begin{proof}
Statement \eqref{item:neighbourhood_containment_is} follows trivially from the fact that $G$ contains no triangle.
To show statement \eqref{item:neighbourhood_containment_AiBi}, assume for a contrary that $uw \notin E$ for some $u \in A_i$ and some $w \in B_i$.
Since $u \in A_i$, we have $N(u) \cap C =\{c_{i-1}, c_{i+1}\}$, and since $w \in B_i$ we have $N(w) \cap C =\{c_i\}$.
Hence, the set $\{c_{i-2},u,c_i,c_{i+2},c_{i-1},w,c_{i+1}\}$ induces an $X_2$ in $G$, 
which cannot be the case.

To prove statements \eqref{item:neighbourhood_containment_Ai} and \eqref{item:neighbourhood_containment_Bi}, consider a graph
$G$ induced by the set $U \cup W$, where
$$U = A[i-2,i+2] \quad \text{and} \quad W = B[i-2,i+2].$$
Since any edge with two endpoints in $U$ (or two endpoints in $W$) could be
extended by some vertices from $\{c_{i-2},c_{i-1},c_{i},c_{i+1},c_{i+2}\}$ to
an odd cycle of size~$\leq 7$ in $G$, the graph $G[U \cup W]$ is bipartite with bipartition classes $U$ and $W$.

To see that \eqref{item:neighbourhood_containment_Ai} holds, first note that
$B_i \subseteq N(u)$ by statement \eqref{item:neighbourhood_containment_AiBi}.
Therefore, suppose that $u$ has a neighbor $v$ in the set $V \setminus (U \cup W)$.

Consider the case when $v \in A_j$ for some $j \notin [i-2,i+2]$.
Since $N(u) \cap C = \{c_{i-1},c_{i+1}\}$ and $N(v) \cap C = \{c_{j-1},c_{j+1}\}$, $u,v$ and the vertices between $c_{j+1}$ and $c_{i-1}$ in $C$ as well as $u,v$ and the vertices between $c_{i+1}$ and $c_{j-1}$ in $C$ 
induce cycles in $G$ of size $\leq m-2$.
Since $|C| \geq 10$, at least one from these cycles has size $\geq 6$, and as such it can not occur in $G$.

So suppose $v \in B_j$ for some $j \notin [i-2,i+2]$.
Since $N(v) \cap C = \{c_j\}$, $u,v$ and the vertices between $c_{i+1}$ and $c_j$ in $C$ as well as $u,v$ and the vertices between $c_{j}$ and $c_{i-1}$ in $C$ induce holes in $G$ of size $\leq m-2$, 
and as such they can not occur in $G$.
So, $N(u) \subseteq W$, which completes the proof of statement \eqref{item:neighbourhood_containment_Ai}.

Statement \eqref{item:neighbourhood_containment_Bi} is proved by analogous argument.
\end{proof}
 
Proposition~\ref{prop:neighbourhoodd_containment} asserts that all the neighbors of the vertices from $A_i$ and from $B_i$ are contained in the set
$B[i-2,i+2]$ and $A[i-2,i+2]$, respectively. 
The next proposition describes the relations that hold between the neighborhoods of the vertices from
$B[i-2,i+2]$ restricted to the set $A_i$ and between the neighborhoods of the vertices from
$A[i-2,i+2]$ restricted to the set $B_i$.

\begin{proposition}
\label{prop:neighbourhoods_in_AiBi}
Let $i \in [m-1]$. For $(i \pm 2,i \pm 1) \in \{(i-2,i-1),(i+2,i+1)\}$, the following hold:
\begin{enumerate}[(1)]
 \item \label{item:neighbourhoods_in_Ai} For every $w,w' \in B_{i \pm 2} \cup A_{i \pm 1}$ the sets $N(w) \cap A_i$ and $N(w') \cap A_i$ are comparable.
 
 Moreover, if $w \in B_{i \pm 2}$ and $w' \in A_{i \pm 1}$, then $N(w) \cap A_i \subseteq N(w') \cap A_i$.
 \item \label{item:neighbourhoods_in_Bi} For every $u,u' \in A_{i \pm 2} \cup B_{i \pm 1}$ the sets $N(u) \cap B_i$ and $N(u') \cap B_i$ are comparable.
 
 Moreover, if $u \in A_{i \pm 2}$ and $u' \in B_{i \pm 1}$, then $N(u) \cap B_i \subseteq N(u') \cap B_i$.
\end{enumerate}
\end{proposition}
\begin{proof}
To prove \eqref{item:neighbourhoods_in_Ai}, we consider the case $(i \pm 2,i \pm 1) =(i-2,i-1)$, as the other one follows by symmetry. Suppose that $w,w' \in B_{i-2} \cup A_{i-1}$ are such that neither 
$N(w) \cap A_i \subseteq N(w') \cap A_i$ nor $N(w') \cap A_i \subseteq N(w) \cap A_i$ holds.
It means that there are $u,u' \in A_i$
such that $wu \in E$, $w'u' \in E$, $wu' \notin E$, and $w'u \notin E$.
Since $w,w' \in B_{i-2} \cup A_{i-1}$, we have $c_{i-2}w,c_{i-2}w' \in E$ and 
$c_{i-4}w, c_{i-3}w, c_{i-4}w', c_{i-3}w' \notin E$.
Furthermore, $ww'\notin E$ and $uu'\notin E$ as $G$ contains no triangle.
Consequently, the set $\{c_{i-3},w,w', c_{i-4},c_{i-2},u,u'\}$ induces a copy of $T_2$ in $G$, which cannot be the case. 
Moreover, if $w \in B_{i-2}$, $w' \in A_{i-1}$, then since $c_i \in (N(w') \cap A_i) \setminus (N(w) \cap A_i)$, the latter statement holds.

To show \eqref{item:neighbourhoods_in_Bi}, we again only consider the case $(i \pm 2,i \pm 1) =(i-2,i-1)$. Suppose that $u,u' \in A_{i-2} \cup B_{i-1}$ are such that neither $N(u') \cap B_i \subseteq N(u) \cap B_i$
nor $N(u) \cap B_i \subseteq N(u') \cap B_i$ holds. 
It means that there are $w,w' \in B_i$ such that $uw,u'w' \in E$ and $u'w, uw' \notin E$.
Since $u,u' \in A_{i-2} \cup B_{i-1}$, we have $uc_{i-1},u'c_{i-1} \in E$ and $uc_{i+1},u'c_{i+1} \notin E$. 
Furthermore, $uu'\notin E$ and $ww'\notin E$ as $G$ contains no triangle.
Hence, the set $\{c_{i-1}, w, w', c_{i+1},u,c_i,u'\}$ induces a copy of $X_3$ in $G$, which cannot be the case.

To see the second part of the statement, assume that
$N(u) \cap B_i \nsubseteq N(u') \cap B_i$ for some $u \in A_{i-2}$, $u' \in B_{i-1}$.
That is, there is $w \in B_i$ such that $uw \in E$ and $u'w \notin E$.
In particular, it means that $u\neq c_{i-2}$.
Note that $uc_{i-1},u'c_{i-1} \in E$.
Consequently, the set $\{c_{i-3},c_{i-2},c_{i-1},c_i,u,u',w\}$ induces a copy of $X_3$ in $G$, 
which is a contradiction.
\end{proof}

Proposition~\ref{prop:neighbourhoods_in_AiBi} allows us to order vertices of $A_i$ based on two properties. We now define relation $<_{A_i}$ which combines them and we show that $<_{A_i}$ is a partial order (see Figure~\ref{fig:neigbourhoodsinAi} for an illustration).
We define for every $u,u' \in A_i$:
\[
u <_{A_i} u' \text{ iff }
\begin{array}{l}
\text{there is $w \in B_{i-2} \cup A_{i-1}$ such that $u \in N(w)$ and $u' \notin N(w)$, or} \\
\text{there is $w \in A_{i+1} \cup B_{i+2}$ such that $u' \in N(w)$ and $u \notin N(w)$,}
 \end{array}
\]
Similarly, we define a relation ${<_{B_i}}$ for every $w,w' \in B_i$:
\[
w <_{B_i} w' \text{ iff }
\begin{array}{l}
\text{there is $u \in A_{i-2} \cup B_{i-1}$ such that $w \in N(u)$ and $w' \notin N(u)$, or} \\
\text{there is $u \in B_{i+1} \cup A_{i+2}$ such that $w' \in N(u)$ and $w \notin N(u)$.}
 \end{array}
\]
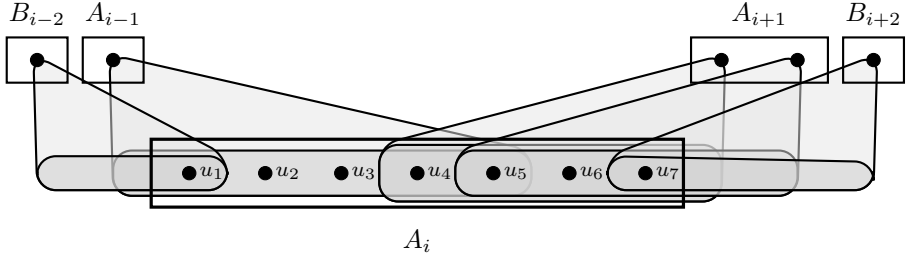
\begin{figure}[htp!]
\centering
\begin{tikzpicture}[xscale=1,yscale=1.5]
\coordinate (w1) at (0,1) {};
\coordinate (lB_2) at (0,1.4) {};
\coordinate (w2) at (1,1) {};
\coordinate (lA_1) at (1,1.4) {};
\coordinate (u1) at (2,0) {};
\coordinate (u2) at (3,0) {};
\coordinate (u3) at (4,0) {};
\coordinate (u4) at (5,0) {};
\coordinate (u5) at (6,0) {};
\coordinate (u6) at (7,0) {};
\coordinate (u7) at (8,0) {};
\coordinate (w3) at (9,1) {};
\coordinate (lA1) at (9.5,1.4) {};
\coordinate (w4) at (10,1) {};
\coordinate (w5) at (11,1) {};
\coordinate (lB2) at (11,1.4) {};

\coordinate (lu1) at (2.3,0) {};
\coordinate (lu2) at (3.3,0) {};
\coordinate (lu3) at (4.3,0) {};
\coordinate (lu4) at (5.3,0) {};
\coordinate (lu5) at (6.3,0) {};
\coordinate (lu6) at (7.3,0) {};
\coordinate (lu7) at (8.3,0) {};

\coordinate (lAi) at (5,-0.6) {};

\begin{scope}[fill opacity=0.5]

\draw[rounded corners=7, fill=gray!20, thick] (6.5,0.2) --(6.5,-0.2) -- (1,-0.2) -- (0.95,1.05) --cycle;
\draw[rounded corners=7, fill=gray!40, thick] (6.5,0.2) --(6.5,-0.2) -- (1,-0.2) -- (1,0.2) --cycle;

\draw[rounded corners=7, fill=gray!20, thick] (2.5,0.15) --(2.5,-0.15) -- (0,-0.15) -- (-0.05,1.05) --cycle;
\draw[rounded corners=7, fill=gray!40, thick] (2.5,0.15) --(2.5,-0.15) -- (0,-0.15) -- (0,0.15) -- cycle;

\draw[rounded corners=7, fill=gray!20, thick] (4.5,0.25) --(4.5,-0.25) -- (9,-0.25) -- (9.05,1.05) --cycle;
\draw[rounded corners=7, fill=gray!40, thick] (4.5,0.25) --(4.5,-0.25) -- (9,-0.25) -- (9,0.25) --cycle;

\draw[rounded corners=7, fill=gray!20, thick] (5.5,0.2) --(5.5,-0.2) -- (10,-0.2) -- (10.05,1.05) --cycle;
\draw[rounded corners=7, fill=gray!40, thick] (5.5,0.2) --(5.5,-0.2) -- (10,-0.2) -- (10,0.2) --cycle;

\draw[rounded corners=7, fill=gray!20, thick] (7.5,0.15) --(7.5,-0.15) -- (11,-0.15) -- (11.05,1.05) --cycle;
\draw[rounded corners=7, fill=gray!40, thick] (7.5,0.15) --(7.5,-0.15) -- (11,-0.15) -- (11,0.1) --cycle;

\end{scope}

\draw[very thick] (1.5,-0.3) --(1.5,0.3) -- (8.5,0.3) -- (8.5,-0.3) --cycle;
\draw[thick] (-0.4,0.8) --(-0.4,1.2) -- (0.4,1.2) -- (0.4,0.8) --cycle;
\draw[thick] (0.6,0.8) --(0.6,1.2) -- (1.4,1.2) -- (1.4,0.8) --cycle;

\draw[thick] (8.6,0.8) --(8.6,1.2) -- (10.4,1.2) -- (10.4,0.8) --cycle;
\draw[thick] (10.6,0.8) --(10.6,1.2) -- (11.4,1.2) -- (11.4,0.8) --cycle;

\tikzstyle{every node}=[circle,minimum size=5pt,inner sep=0pt,draw,fill]
\node at (w1) {};
\node at (w2) {};
\node at (w3) {};
\node at (w4) {};
\node at (w5) {};
\node at (u1) {};
\node at (u2) {};
\node at (u3) {};
\node at (u4) {};
\node at (u5) {};
\node at (u6) {};
\node at (u7) {};

\tikzstyle{every node}=[inner sep=2pt]
\begin{footnotesize}
\node at (lu1) {$u_1$};
\node at (lu2) {$u_2$};
\node at (lu3) {$u_3$};
\node at (lu4) {$u_4$};
\node at (lu5) {$u_5$};
\node at (lu6) {$u_6$};
\node at (lu7) {$u_7$};
\end{footnotesize}
\node at (lAi) {$A_i$};
\node at (lB_2) {$B_{i-2}$};
\node at (lA_1) {$A_{i-1}$};
\node at (lA1) {$A_{i+1}$};
\node at (lB2) {$B_{i+2}$};

\end{tikzpicture}

\caption{\label{fig:neigbourhoodsinAi} The neighborhoods of the vertices from $B_{i-2} \cup A_{i-1} \cup A_{i+1} \cup B_{i+2}$ restricted to $A_i$. We have $u_1 <_{A_i} \{u_2,u_3\} <_{A_i} u_4 <_{A_i} u_5 <_{A_i} u_6 <_{A_i} u_7$.}

\end{figure}
\begin{proposition} The following statements hold for every $i \in [m-1]$:
\label{prop:weak_orders}
\begin{enumerate}
 \item \label{item:weak_order_Ai} $(A_i,{<_{A_i}})$ is a strict partial order.
 Moreover, $u,u' \in A_i$ are incomparable in $(A_i,{<_{A_i}})$ if and only if $N(u) = N(u')$. 
 \item \label{item:weak_order_Bi} $(B_i,{<_{B_i}})$ is a strict partial order.
 Moreover, $w,w' \in B_i$ are incomparable in $(B_i,{<_{B_i}})$ if and only if $N(w) = N(w')$. 
\end{enumerate}
\end{proposition}
\begin{proof}
Let $i \in [m-1]$ be fixed.
To prove that $(A_i,<_{A_i})$ is a strict partial order, we need to show that $<_{A_i}$ is irreflexive and transitive. The irreflexivity follows from the definition, in aim to show transitivity, we first prove that $<_{A_i}$ is antisymmetric.
Suppose for a contrary that there are $u,u' \in A_i$ such that $u <_{A_i} u'$, and $u' <_{A_i} u$.
Suppose $u <_{A_i} u'$ is witnessed by a vertex $w \in B_{i-2} \cup A_{i-1}$ such that $u \in N(w)$ and  $u' \notin N(w)$; the other case $w \in A_{i+1} \cup B_{i+2}$ is analogous.
By Proposition~\ref{prop:neighbourhoods_in_AiBi}.\eqref{item:neighbourhoods_in_Ai}, there is no $w' \in B_{i-2} \cup A_{i-1}$ such that $u' \in N(w')$ and $u \notin N(w')$.
Hence, since $u' <_{A_i} u$, there must be a vertex $w' \in A_{i+1} \cup B_{i+2}$ such that
$u \in N(w')$ and $u' \notin N(w')$.
We have $uc_{i+1}, u'c_{i+1} \in E$ and $uc_{i+2},u'c_{i+2}, uc_{i+3},u'c_{i+3} \notin E$ as $\{u,u'\} \subseteq A_i$.
We have also $wc_{i+1},wc_{i+2},wc_{i+3},w'c_{i+1},w'c_{i+3} \notin E$ and $w'c_{i+2} \in E$ as $w \in B_{i-2} \cup A_{i-1}$ and $w' \in A_{i+1} \cup B_{i+2}$.
Moreover, $uu', ww'\notin E$, by Proposition~\ref{prop:neighbourhoodd_containment}. 
Consequently, the set $\{w,w',c_{i+1},c_{i+3}, u,u',c_{i+2}\}$ induces a copy of $X_2$ in $G$, which cannot be the case.

To show transitivity, suppose for a contrary that there are vertices $u,u',u''\in A_i$ such that $u <_{A_i} u'$ and $u' <_{A_i} u''$, but $u <_{A_i} u''$ does not hold.
Suppose $u <_{A_i} u'$ is witnessed by a vertex $w \in B_{i-2} \cup A_{i-1}$ 
such that $u \in N(w)$ and $u' \notin N(w)$; the other case $w \in B_{i+1} \cup A_{i+2}$ is symmetric.
We have $u'' \in N(w)$ as otherwise $u <_{A_i} u''$, by definition.
Suppose $u' <_{A_i} u''$ is witnessed by a vertex $w' \in B_{i-2} \cup A_{i-1} \cup A_{i+1} \cup B_{i+2}$.
Note that if $w' \in B_{i-2} \cup A_{i-1}$, then $u'\notin N(w')$ and $u' \in N(w')$, which enforces also $u \in N(w')$ as $u <_{A_i} u'$ and we already proved that $<_{A_i}$ is antisymmetric.
Thus, $u \in N(w')$ and $u'' \notin N(w')$, which shows $u <_{A_i} u''$.
Hence, we must have $w' \in A_{i+1} \cup B_{i+2}$, and so $u' \notin N(w')$ and $u'' \in N(w')$.
Moreover, $u \in N(w')$ as otherwise $u <_{A_i} u''$.
As $\{u',u''\} \subseteq A_i$, $w \in A_{i-2} \cup B_{i+1}$, and $w' \in A_{i+1} \cup B_{i+2}$, we have $uu', ww'\notin E$, by Proposition~\ref{prop:neighbourhoodd_containment}. 
Consequently, $\{w, w',c_{i+1},c_{i+3},u',u'',c_{i+2}\}$ induces a copy of $X_2$ in $G$, which is not possible. 
We conclude that $(A_i,<_{A_i})$ is a strict partial order.

By definition,  if $N(u)=N(u')$, then $u$ and $u'$ are incomparable in $(A_i,<_{A_i})$.
Hence, for the second statement of \eqref{item:weak_order_Ai}, it is enough to show that  $N(u)\neq N(u')$ implies that $u$ and $u'$ are comparable in $(A_i,<_{A_i})$.
Let $w$ be a vertex such that $wu \in E$ and $wu' \notin E$. By Proposition \ref{prop:neighbourhoodd_containment}.(\ref{item:neighbourhood_containment_AiBi}) and (\ref{item:neighbourhood_containment_Ai}), $w \in B_{i-2} \cup A_{i-1} \cup A_{i+1} \cup B_{i+2}$.
However, if $w \in B_{i-2} \cup A_{i-1}$ then $u <_{A_i} u'$ and if $w \in A_{i+1} \cup B_{i+2}$ then $u' <_{A_i} u$, by definition.
Hence, $u$ and $u'$ are comparable in~$<_{A_i}$.

The proof of \eqref{item:weak_order_Bi} is similar. For antisymmetry, suppose that we have $w, w' \in B_i$ such that $w <_{B_i} w'$ and $w' <_{B_i} w$. Let $w <_{B_i} w'$ and $w' <_{B_i} w$ be witnessed by $u$ and $u'$ from $A_{i-2} \cup B_{i-1} \cup B_{i+1} \cup A_{i+2}$, respectively.
Analogously to (1), by Proposition~\ref{prop:neighbourhoods_in_AiBi}.\eqref{item:neighbourhoods_in_Bi}, we can assume that $u \in A_{i-2} \cup B_{i-1}$ and $u' \in B_{i+1} \cup A_{i+2}$ and $uw,u'w \in E$, $uw',u'w' \notin E$. 
Observe that
 the set $\{c_{i-1},w,w',c_{i+1},u,c_i,u'\}$ induces a copy of $X_3$ in $G$,
a contradiction.

For transitivity of $<_{B_i}$, suppose that for some $w,w'$, and $w'' \in B_i$ we have $w <_{B_i} w'$ and $w' <_{B_i} w''$, but $w <_{B_i} w''$ does not hold. 
By symmetry of the proof of \eqref{item:weak_order_Ai}, we reach the case $u \in A_{i-2} \cup B_{i-1}$ and $u' \in B_{i+1} \cup A_{i+2}$, 
$uw, uw'',u'w,u'w'' \in E$ and $uw',u'w' \notin E$. 
Then, one can easily check that the set $\{w',w'',c_{i+1},u,c_i,u',c_{i+2}\}$ induces a copy of $X_2$ in $G$, a contradiction.

Now, assume that $N(w)\neq N(w')$. Without 
loss of generality assume that there exists $u \in A_{i-2} \cup B_{i-1} \cup B_{i+1} \cup A_{i+2}$ such that $u \in N(w) \setminus N(w')$. Analogously as before, observe that if $u \in A_{i-2} \cup B_{i-1}$ then $w <_{B_i} w'$ and if $u \in B_{i+1} \cup A_{i+2}$ then $w' <_{B_i} w$.
Therefore, $w$ and $w'$ are comparable, which finishes the proof.
\end{proof}

Finally, for every $i \in [m-1]$ we order arbitrarily the elements inside every antichain of $(A_i,{<_{A_i}})$ and of $(B_i,{<_{B_i}})$, obtaining strict
linear orders $(A_i,{<_{A_i}})$ and $(B_{i},{<_{B_i}})$. 
We introduce a binary relation $\prec$ defined on the set $V$, such that $v \prec v'$ for $v,v' \in V$ if one of the following conditions holds for some $i \in [m-1]$:
\begin{itemize}
 \item $v,v' \in A_i$, $v <_{A_i} v'$, and $v,v'$ are consecutive in $(A_i,{<_{A_i}})$,
 \item $v,v' \in B_i$, $v <_{B_i} v'$, and $v,v'$ are consecutive in $(B_i,{<_{B_i}})$,
 \item $v$ is the maximum of $(A_i,{<_{A_i}})$ and $v'$ is the minimum of $(B_{i+1},{<_{B_{i+1}}})$,
 \item $v$ is the maximum of $(B_i,{<_{B_i}})$ and $v'$ is the minimum of $(A_{i+1},{<_{A_{i+1}}})$.
\end{itemize}
Informally, to get an embedding of $G$ into a cylinder (the shortest hole is even) or into a M\"obius strip (the shortest hole is odd) which locally satisfies the adjacency and the
enclosure properties, we place the vertices $v,v'$ satisfying $v \prec v'$ next to each other, $v$  before $v'$ assuming that the border of the cylinder or the M\"obius strip 
are oriented as shown in Figure~\ref{fig:cylinder}.
In what follows we extend ${\prec}$ relation as follows:

\begin{itemize}
\item For every $V' \subsetneq V$ by ${<_{V'}}$ we denote the transitive closure of ${\prec}$ restricted to~$V'$,
\item For $v,v' \in V$ we set $v <_{cl} v'$ if $v,v' \in A[i-2,i+2]$ and $v <_{A[i-2,i+2]} v'$ for some $i \in [m-1]$ 
or $v,v' \in B[i-2,i+2]$ and $v <_{B[i-2,i+2]} v'$ for some $i \in [m-1]$.
\end{itemize}

Finally, the following lemma characterizes the global structure of an almost bipartite permutation graph. 

\begin{lemma}%
\label{lem:slices_are_bipartite_permutation}
Let $i,j$ be such that $i \leq j$, $|j-i| = m-3$. Let $U = A[i,j]$ and $W=B[i,j].$ 
Then $G[U \cup W]$ is a bipartite permutation graph with bipartition classes $U$ and~$W$.

Moreover, $(U,{<_U})$ and $(W,{<_W})$ are strict linear orders that satisfy the adjacency and enclosure properties in $G[U \cup W]$.
\end{lemma}
\begin{proof}
Proposition~\ref{prop:neighbourhoodd_containment} asserts there is no edge between a vertex in $V[j-1,j]$ and a vertex in $V[i,i+1]$. 
In particular, $G[U \cup W]$ is a bipartite graph and $(U,{<_U})$ and $(W,{<_W})$ are strict linear orders.  
Given Theorem \ref{thm:bip_char}.(\ref{thm:bip_char_3}), to prove the lemma we need to show that $(W,{<_W})$ satisfies the adjacency and enclosure properties in $G[U \cup W]$.

To prove the adjacency property, consider $u\in A_k\subseteq U$ for some suitable $k$. 
Recall that by Proposition~\ref{prop:neighbourhoodd_containment}.\eqref{item:neighbourhood_containment_Ai}, $B_k\subseteq N(u) \subseteq B[k-2,k+2]$. To show that $N(u)$ consists of consecutive vertices in $W$ it suffices to note that:
\begin{itemize}
    \item if $w\in A_{k+1}$, $w'\in B_{k+2}$ and $uw'\in E$ then $uw\in E$, by Proposition~\ref{prop:neighbourhoods_in_AiBi},
    \item if $w,w'\in A_{k+1}$ (resp. $w,w'\in B_{k+2}$) are such that $w <_{A_{k+1}} w'$ (resp. $w <_{B_{k+2}} w'$) and $uw'\in E$, then $uw\in E$, by Proposition~\ref{prop:weak_orders},
\end{itemize}
and that analogous statements hold by symmetry for the part of $N(u)$ contained in $A_{k-1}\cup B_{k-2}$. If $u\in B_k\subseteq U$, the case analysis is similar (one needs to swap letters $A$ and $B$ in the reasoning above). Therefore, the adjacency property is proved.
 
To show that $(W,{<_W})$ satisfies the enclosure property assume for a contradiction that there are $w,w',w'' \in W$ and $u,u' \in U$ such that $N(u')\subseteq N(u)$, $w <_{W} w' <_{W} w''$ and $uw,uw',uw'' \in E$,
$u'w' \in E$, and $u'w,u'w'' \notin E$.

\begin{claim}
There is $k \in [i,j]$ %
such that  either $u,u' \in A_k$, or $u,u' \in B_k$.
\end{claim}
\begin{proof}[Proof of Claim]
If $u\in B_{k}$, then since $N(u') \cap C \subseteq N(u) \cap C = \{c_k\}$, we have $u' \in B_k$, so the claim holds.
Therefore, assume that $u\in A_{k}$, and suppose that $u' \notin A_k$. 
Then $N(u') \cap C \subseteq N(u) \cap C = \{c_{k-1},c_{k+1}\}$.
Assuming $u <_{U} u'$ (the other case is symmetric), we have that $u' \in B_{k+1}$.
Due to Proposition~\ref{prop:neighbourhoodd_containment} and $w' <_W w''$ we have $w',w'' \in A[k-1,k+2]$. 
Moreover, as we already proved that 
$N(u')$ is consecutive in $(W,<_W)$ (adjacency property), 
and $c_{k+1}\in N(u')$, we have $c_{k+1}<_W w''$. Therefore $w''\in A_{k+1}\cup B_{k+2}$.
Note that:
\begin{itemize}
    \item if $w''\in B_{k+2}$, then since $uw''\in E$, by Proposition~\ref{prop:neighbourhoods_in_AiBi}.\eqref{item:neighbourhoods_in_Bi}, we have $u'w''\in E$, a contradiction,
    \item if $w'' \in A_{k+1}$, then by Proposition~\ref{prop:neighbourhoodd_containment} we would have $u'w''\in E$, a contradiction. 
   \qedhere
\end{itemize}
\end{proof}

Suppose $u,u' \in A_k$.
Since $u'w, u'w'' \notin E$ and $u'c_{k-1}, u'c_{k+1}\in E$, we have by adjacency property of $(W,<_{W})$ that $w<_{W}c_{k-1}<_W w''$.
Therefore, we must have $w \in B_{k-2} \cup A_{k-1}$ and $w'' \in A_{k+1} \cup B_{k+2}$.
Observe that $w''$ witnesses that $u'<_{A_k} u$ by definition, however, $w$ witnesses the opposite, that is $u<_{A_k} u'$. 
We have a contradiction by Proposition~\ref{prop:weak_orders}.

Suppose $u,u' \in B_k$.
An analysis, which is analogous to the one in the previous case (again, it is enough to swap letters $A$ and $B$ in that reasoning above), gives us that we must have $w \in A_{k-2} \cup B_{k-1}$ and $w'' \in B_{k+1} \cup A_{k+2}$.
Again, we obtain a contradiction by the definition of $<_{B_k}$ and  Proposition~\ref{prop:weak_orders}.
\end{proof}

Lemma~\ref{lem:slices_are_bipartite_permutation} provides an interesting view on classification of almost bipartite permutation graphs. Specifically, if %
$m$ is even,
then the graph may be drawn on a cylinder, whose boundary consists of two closed curves, one of which traverses vertices of $A[0,m-1]$, and the second one---the vertices of $B[0,m-1]$. If in turn %
$m$ is odd,
then the graph can be represented on a M\"obius strip, whose boundary traverses consecutive vertices of $A[0,m-1]$ and then $B[0,m-1]$ (recall Figure~\ref{fig:cylinder}).

The following definitions are taken from \cite{HV13}. 
A \emph{hole cut} of $G$ is a vertex set $X \subseteq V$ such that $G - X$ is a bipartite permutation graph.
Lemma \ref{lem:slices_are_bipartite_permutation} asserts that for every $i \in [m-1]$ the set $V[i,i+1]$ is a hole cut in $G$. 
A hole cut $X$ of $G$ is \emph{minimum} if $G$ does not have a hole cut whose size is strictly smaller than the size of $X$.
A~hole cut $X$ of $G$ is \emph{minimal} if any proper subset of $X$ is not a hole cut in~$G$.

The next proposition describes the structure of every hole in $G$.
\begin{proposition}\label{prop:holes}
Suppose $C'$ is a hole of size $k$ in $G$ for some $k \geq m$. 
Then, the consecutive 
vertices of $C'$ can be labeled by $c'_0,c'_1,\ldots,c'_{k-1}$ so as the following conditions hold (the indices are taken modulo $k$):
\begin{itemize}
 \item $c'_ic'_{i+1} \in E$ for every $i \in [k-1]$,
 \item $c'_{i} <_{cl} c'_{i+2}$ for every $i \in [k-1]$,
 \item $C' \cap \{c'': c'_i <_{cl} c'' <_{cl} c'_{i+2}\} = \emptyset$ for every $i \in [k-1]$. 
\end{itemize}
\end{proposition}
\begin{proof}
Let $c'_0=v_i$, $c'_1$, $c'_2$, $\dots$, $c'_{n-1}$ be consecutive vertices of $C'$ denoted in such a way that $c'_0<_{cl} c'_2$. %
We note that we can suppose it as $c'_0, c'_2\in N(c'_1)$, thus, by Proposition~\ref{prop:neighbourhoodd_containment}.(\ref{item:neighbourhood_containment_Ai})~and~(\ref{item:neighbourhood_containment_Bi}), both $c'_0,c'_2$ belong to $A[\ell-2,\ell+2]$ or both belong to $B[\ell-2,\ell+2]$ for some $\ell\in[m-1]$.

Now, we show that if there exists $j\in[m-1]$ such that $c'_{j}<_{cl} c'_{j+2}$, then $c'_{j+1}<_{cl} c'_{j+3}$.
Suppose, for contradiction that $c'_{j}<_{cl} c'_{j+2}$ and $c'_{j+1}\nless_{cl} c'_{j+3}$. Let $i\in [m-1]$ be such that $c_{j+2}\in (A_i\cup B_i)$.
Similarly, as $c'_{j+1},c'_{j+3}\in N(c'_{j+2})$, either $c'_{j+1}, c'_{j+3} \in B[i-2,i+2]$ if  $c'_{j+2}\in A_i$ or $c'_{j+1},c'_{j+3}\in A[i-2,i+2]$ if $c'_{j+2}\in B_i$.
In both cases Lemma \ref{lem:slices_are_bipartite_permutation} implies that $<_{cl}$ restricted to $V[i-4,i+2]$ is a strong ordering of $G[V[i-4,i+2]]$.
Moreover, $c'_{j+1},c'_{j+3}$ are comparable in $<_{cl}$, by Proposition~\ref{prop:neighbourhoodd_containment}.(\ref{item:neighbourhood_containment_Ai})~and~(\ref{item:neighbourhood_containment_Bi}), thus, $c'_{j+3}<_{cl} c'_{j+1}$.
Therefore Theorem~\ref{thm:bip_char}.(\ref{thm:bip_char_2}) yields $c'_jc'_{j+3}\in E$, so a chord in~$C'$---contradiction. 
Therefore $c'_{j}<_{cl} c'_{j+2}$ implies $c'_{j+1}<_{cl} c'_{j+3}$ for every integer $j$.
Applying above observation repeatedly for $j=0,1,2,\ldots$, we get that $c'_0<_{cl} c'_2<_{cl} c'_4<_{cl}\ldots$ and $c'_1<_{cl} c'_3<_{cl} c'_5<_{cl}\ldots$ 

Suppose for the sake of contradiction that there exists $j\notin\{i,i+2\}$ such that $c'_i <_{cl} c'_j <_{cl} c'_{i+2}$. Then by Lemma~\ref{lem:slices_are_bipartite_permutation}, $c'_jc'_{i+1}\in E$ due to the adjacency property. But the edge $c'_jc'_{i+1}$ is a chord in $C'$. This completes the proof.%
\end{proof}
The structure of holes described above asserts that for every $i \in [m-1]$ the sets
$A[i-2,i+2]$ and $B[i-2,i+2]$ are hole cuts.
We use this observation to prove the following statement about minimal hole cuts in $G$.
\begin{proposition}\label{prop_min_cut}
Every minimal hole cut $X$ in $G$ is fully contained in the set
$V[i-2,i+2]$ for some $i \in [m-1]$. 
\end{proposition}
\begin{proof}
First, note that we can choose elements $z_1,x_1,x_2,z_2$ in $V$ and an index $i \in [m-1]$ such that the following conditions hold:
\begin{itemize}
 \item we have $z_1 \prec x_1 \leq_{cl} x_2 \prec z_2$, the set $X' = \{x: x_1 \leq_{cl} x \leq_{cl} x_2\}$ is non-empty and is contained in $X$, and the elements $z_1,z_2$ are not in $X$,
 \item the set $X' \cup \{z_1,z_2\}$ is contained in either $B[i-2,i+2]$ or in $A[i-2,i+2]$
 and we have $z_1 \in B[i-2,i]$ and $z_2 \in B[i,i+2]$ if 
$X' \cup \{z_1,z_2\} \in B[i-2,i+2]$, and similarly for the other case.
\end{itemize}
Note that such a choice of $z_1,x_1,x_2,z_2$ and $i$ is possible as the sets $A[j,j+3]$ and $B[j,j+3]$ are hole cuts for every $j \in [m-1]$, by combining Proposition~\ref{prop:holes} and Proposition~\ref{prop:neighbourhoodd_containment}.(\ref{item:neighbourhood_containment_Ai}),(\ref{item:neighbourhood_containment_Bi}).
For the rest of the proof we assume $X' \subset B[i-2,i+2]$, $z_1 \in B[i-2,i]$ and $z_2 \in B[i,i+2]$ (see Figure \ref{fig:minimal_hole_cut} for an illustration).

\begin{figure}[htp!]
\centering
\begin{tikzpicture}[xscale=2,yscale=1.5]
\coordinate (lcp0) at (-0.9,1) {};
\coordinate (cp0) at (-0.75,1) {};

\coordinate (cp1) at (0,0) {};
\coordinate (lcp1) at (0.25,0) {};

\coordinate (cp2) at (0.75,1) {};
\coordinate (lcp2) at (1,1) {};

\coordinate (cp3) at (1.25,0) {};
\coordinate (cp4) at (2,1) {};

\coordinate (cp5) at (2.75,0) {};

\coordinate (cp6) at (3.25,1) {};
\coordinate (lcp6) at (3.1,1) {};
\coordinate (cp7) at (4,0) {};
\coordinate (lcp7) at (4.25,0) {};
\coordinate (cp8) at (4.75,1) {};
\coordinate (lcp8) at (5.00,1) {};

\coordinate (z1) at (-0.15,1) {};
\coordinate (x1) at (0.15,1) {};
\coordinate (lz1) at (-0.3,1.0) {};
\coordinate (lx1) at (0.3,1.0) {};

\coordinate (x2) at (3.85,1) {};
\coordinate (z2) at (4.15,1) {};
\coordinate (lx2) at (3.7,1.0) {};
\coordinate (lz2) at (4.3,1) {};

\coordinate (x) at (2,1) {};
\coordinate (lx) at (2.15,1) {};

\coordinate (y) at (2,0) {};
\coordinate (ly) at (2.12,-0.07) {};

\coordinate (lB0) at (0.0,-0.4) {};
\coordinate (lA0) at (0.0,1.4) {};
\coordinate (lB1) at (1.0,-0.4) {};
\coordinate (lA1) at (1.0,1.4) {};
\coordinate (lB2) at (2.0,-0.4) {};
\coordinate (lA2) at (2.0,1.4) {};
\coordinate (lB3) at (3.0,-0.4) {};
\coordinate (lA3) at (3.0,1.4) {};
\coordinate (lB4) at (4.0,-0.4) {};
\coordinate (lA4) at (4.0,1.4) {};

\path (cp0) edge[thick, dashed] (cp1);
\path (cp1) edge[thick, dashed] (cp2);
\path (cp2) edge[thick, dashed] (cp3);
\path (cp3) edge[thick, dashed] (cp4);
\path (cp4) edge[thick, dashed] (cp5);
\path (cp5) edge[thick, dashed] (cp6);
\path (cp6) edge[thick, dashed] (cp7);
\path (cp7) edge[thick, dashed] (cp8);

\path (z1) edge[thick] (y);
\path (z2) edge[thick] (y);
\path (cp1) edge[thick] (z1);
\path (cp7) edge[thick] (z2);

\draw[thick] (-0.45,-0.2)--(-0.45,0.2) -- (0.45,0.2) -- (0.45,-0.2)--cycle;
\draw[thick] (-0.45,0.8)--(-0.45,1.2) -- (0.45,1.2) -- (0.45,0.8)--cycle;

\draw[thick] (0.55,-0.2)--(0.55,0.2) -- (1.45,0.2) -- (1.45,-0.2)--cycle;
\draw[thick] (0.55,0.8)--(0.55,1.2) -- (1.45,1.2) -- (1.45,0.8)--cycle;

\draw[thick] (1.55,-0.2)--(1.55,0.2) -- (2.45,0.2) -- (2.45,-0.2)--cycle;
\draw[thick] (1.55,0.8)--(1.55,1.2) -- (2.45,1.2) -- (2.45,0.8)--cycle;

\draw[thick] (2.55,-0.2)--(2.55,0.2) -- (3.45,0.2) -- (3.45,-0.2)--cycle;
\draw[thick] (2.55,0.8)--(2.55,1.2) -- (3.45,1.2) -- (3.45,0.8)--cycle;

\draw[thick] (3.55,-0.2)--(3.55,0.2) -- (4.45,0.2) -- (4.45,-0.2)--cycle;
\draw[thick] (3.55,0.8)--(3.55,1.2) -- (4.45,1.2) -- (4.45,0.8)--cycle;

\begin{scope}[fill opacity=0.2]
\draw[rounded corners=7, fill=gray!90, thick] (0.0, 0.75)--(0.0,1.25) -- (4,1.25) -- (4,0.75)--cycle;

\end{scope}

\tikzstyle{every node}=[circle,minimum size=5pt,inner sep=0pt,draw,fill]

\node at (z1) {};
\node at (x1) {};
\node at (z2) {};
\node at (x2) {};

\node at (x) {};
\node at (y) {};

\node at (cp0) {};
\node at (cp1) {};
\node at (cp2) {};
\node at (cp3) {};
\node at (cp4) {};
\node at (cp5) {};
\node at (cp6) {};
\node at (cp7) {};
\node at (cp8) {};

\tikzstyle{every node}=[inner sep=2pt]
\begin{footnotesize}
\node at (lcp0) {$c'_{p}$};
\node at (lcp1) {$c'_{p+1}$};
\node at (lcp2) {$c'_{p+2}$};

\node at (lcp6) {$c'_{q}$};
\node at (lcp7) {$c'_{q+1}$};
\node at (lcp8) {$c'_{q+2}$};

\node at (lz1) {$z_1$};
\node at (lx1) {$x_1$};
\node at (lz2) {$z_2$};
\node at (lx2) {$x_2$};

\node at (lx) {$x$};
\node at (ly) {$y$};

\end{footnotesize}

\node at (lA0) {$B_{i-2}$};
\node at (lB0) {$A_{i-2}$};
\node at (lA1) {$A_{i-1}$};
\node at (lB1) {$B_{i-1}$};
\node at (lA2) {$B_{i}$};
\node at (lB2) {$A_{i}$};
\node at (lA3) {$A_{i+1}$};
\node at (lB3) {$B_{i+1}$};
\node at (lA4) {$B_{i+2}$};
\node at (lB4) {$A_{i+2}$};

\end{tikzpicture}

\caption{Illustration of the proof: the cycle $C'$ is marked with a dashed line. The set $X'$ is shaded.}
\label{fig:minimal_hole_cut}
\end{figure}
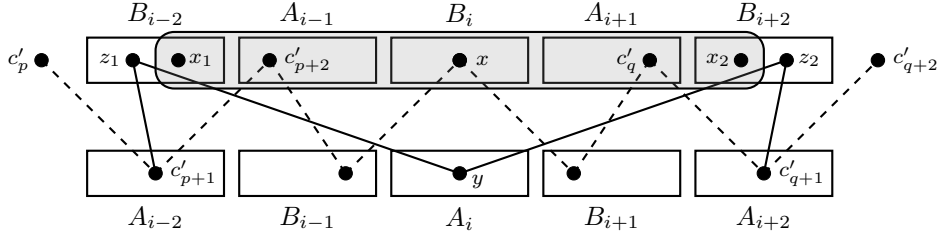

Suppose $Y'$ is the set consisting of all the neighbors of 
$z_1$ and $z_2$; that is, $Y' = N(z_1) \cap N(z_2)$.
Clearly, we have $Y' \subset A[i-2,i+2]$.
To complete the proof of the proposition we show that:
\begin{itemize}
 \item every element of $Y'$ is a member of $X$,
 \item $X' \cup Y'$ is a hole cut in $G$.
\end{itemize}
Then we have $X' \cup Y' = X$ by minimality of $X$ 
and consequently $X \subset V[i-2,i+2]$.
So, it remains to prove the claims about the set $Y'$. 

Suppose we have $y \in Y'$ such that $y \notin X$.
Since $X$ is a minimal hole cut, $X \setminus \{x\}$ is not a hole cut, where $x$ is some fixed element from $X'$.
That is, there is a hole $C'$ in $G - (X \setminus \{x\})$.
Note that $C'$ must contain~$x$.
Suppose $c'_0,\ldots,c'_{\ell-1}$ for some $\ell \geq 9$ are consecutive vertices in
$C'$ chosen such that $c'_j <_{cl} c_{j+2}$ for every $j \in [\ell-1]$ (indices are taken modulo~$\ell$).
Now we pick $p,q \in [\ell-1]$ such that 
$c'_p <_{cl} z_1 \leq _{cl} c'_{p+2}$ and $c'_q \leq_{cl} z_2  <_{cl} c'_{q+2}$.
Since $x \in C'$, we have $c'_{p+2} \leq_{cl} x$ and $c'_{q} \leq_{cl} x$.
Note that $c'_{p+1}$ is adjacent to $z_1$ and $c'_{q+1}$ is adjacent to $z_2$.
Next we replace in $C'$ all the vertices between $c'_{p+2}$ and $c'_q$ (this set includes $x$)  with the vertices $z_1,y,z_2$ and we obtain a cycle $C''$ 
containing no elements from $X$.
Clearly, we can easily find a hole among the elements from $C''$ that avoids all the elements from $X$.
This yields a contradiction as $X$ is a hole cut.

To prove the second claim, suppose there is a hole $C'$ in $G - (X' \cup Y')$.
By Proposition \ref{prop:holes} there are $c'_1,c'_2,c'_3 \in C'$ such that $c'_1 <_{cl} X' <_{cl} c'_3$ and $c'_1,c'_3 \in N(c_2)$.
However, this yields $c'_2 \in Y'$, which is a contradiction.
\end{proof}

\section{Proof of Theorem~\ref{thm_main_fpt}}\label{sec_fpt_proof}

The aim of this section is to provide a complete proof of Theorem~\ref{thm_main_fpt} using structural results from the previous section.
Let us start by showing that the \bpd problem can be decided in polynomial time on almost bipartite permutation graphs.

\begin{lemma}\label{lem_max_flow}
Let $(G,k)$ be an instance of \bpd where $G$ is an $n$-vertex almost bipartite permutation graph.
Then \bpd can be decided in time $\Oh{n^6}$. 
\end{lemma}
\begin{proof}
If $G$ is a bipartite permutation graph, $(G,k)$ is a \yes-instance, thus, we are done in this case. If $G$ is not connected, we can consider each connected component independently and, at the end, we compare $k$ with  
 the total number of deleted vertices over all components.
Let $G'$ be a connected $r$-vertex component of $G$ such that $G'$ is not a bipartite permutation graph (otherwise, clearly, no vertex needs to be deleted).
Let $C=\{c_0,\dots,c_{m-1}\}$ be a shortest hole in $G'$ (it exists as $G'$ is not a bipartite permutation graph). 
It can be found in time $\Oh{r^6}$ as follows.
We iterate over all possible four-element subsets $S=\{v_1,v_2,v_3,v_4\}$ of $V(G')$. 
For these $S$ for which $G'[S]$ is an induced $P_4$, with consecutive vertices $v_1,v_2,v_3,v_4$, we construct a graph $\widetilde{G'}$ by removing the vertices from $(N(v_2) \cup N(v_3)) \setminus \{v_1,v_4\}$ (note that $v_2$ and $v_3$ also get removed).
Then we find a shortest $v_1$-$v_4$-path in $\widetilde{G'}$ in time $\Oh{r^2}$.

By Proposition~\ref{prop_min_cut}, every minimal hole cut $X$ in $G'$ is contained in the set $V'=V_{G'}[i-2,i+2]$ for some $i \in [m-1]$.
Therefore, we may check all the possibilities where a minimal cut is contained.
For every $i$, we run an algorithm for finding a maximum flow in the following digraph $H_i$.

\def\iin{\mathsf{in}}\def\oout{\mathsf{out}}
Digraph $H_{i}$ has the vertex set $V'\times\{\iin,\oout\}\cup\{s,t\}$ and arc set consisting~of:
\begin{itemize}
\item all arcs of the form $(u,\oout)(v,\iin)$, where $uv$ is an edge of $G'[V']$,
\item $s(v,\iin)$ if there exists $u\in V_{G'}[i-4,i-3]$ such that $uv$ is an edge of $G'$,
\item $(u,\oout)t$ if there exists $v\in V_{G'}[i+3,i+4]$ such that $uv$ is an edge of $G'$,
\item $(u,\iin)(u,\oout)$ for all $u\in V'$.
\end{itemize}
Set capacities of arcs of the form $(u,\iin)(u,\oout)$ to $1$ and capacities of all the remaining arcs to $\infty$ (practically $|V_{G'}|$). It is readily seen that minimum $(s,t)$-cut in the defined network $H_i$ corresponds to minimum hole cut in $G'[V']$ (arc of unit capacity $(u,\iin)(u,\oout)$ naturally corresponds to the vertex $u$ of $G'$).

Therefore it remains to apply classical max-flow algorithm to each $H_i$ for $i\in[m-1]$ and 
remember the smallest size $k_{G'}$ of minimal $(s,t)$-cuts.
This can be performed in time 
$\Oh{m\cdot(|V'|+2)\cdot(|E_{G'[V']}|+2|V'|)^2}=\Oh{r^6}$
 \cite{edmonds1972theoretical}.
Finally, $(G,k)$ is a \yes-instance if and only if the sum of remembered sizes $k_{G'}$ over the all considered connected components $G'$ is at most $k$. Clearly, the total running time is $\Oh{n^6}$.
\end{proof}

We now propose the algorithm.
Given an $n$-vertex graph $G=(V,E)$ and number $k$, we want to answer 
the \bpd problem.
 We say that $(G,k)$ is the \emph{initial} instance.
We split our algorithm into two parts.
The first part consists of a branching algorithm for deletion to almost bipartite permutation graphs. The output of the first part is a set of instances $(G',k')$ where $G'$ is an almost bipartite permutation graph and $0\le k'\le k$ (or \no-answer is no such instance exists) such that the initial instance $(G,k)$ is a \yes-instance if and only if at least one of these instances is a \yes-instance.
We show that the overall time of the first phase is $\Oh{n^9\cdot 9^k}$. 
In the second part, the algorithm runs an $\Oh{n^6}$-time algorithm for \bpd for each instance $(G',k')$ output by the first phase.
In the second part, the algorithm runs an $\Oh{n^6}$ algorithm for \bpd for each instance $(G',k')$ output by the first phase.

Let us start with the first part. 
We say that $X \subseteq V$ is a \emph{forbidden set} if $G[X]$ is isomorphic to one of the graphs: $K_3, T_2, X_2, X_3, C_5, C_6, C_7, C_8, C_9$. 
We define the following rule.
\begin{enumerate}[Rule]
\item: Given an instance $(G, k)$, $k\ge 1$, and a minimal forbidden set $X$, branch into $|X|$ instances, $(G - v, k - 1)$ for each $v \in X$. \label{rule_small_forbidden_sets}
\end{enumerate}
Starting with the initial instance, the algorithm applies the rule exhaustively.
In other words, the algorithm is a branching tree with leaves corresponding to instances $(G',k')$ where  $k'=0$ or $G'$ is an almost bipartite permutation graph. 
Clearly, as at least one vertex from each forbidden set must be removed from $G$, the initial instance is a \yes-instance if and only if at least one of the leaves is a \yes-instance.

The algorithm continues to the second part only with such leaves $(G',k')$ that $G'$ is an almost bipartite permutation graph (as otherwise, the leaf is \no-instance). 
It runs the algorithm described in Lemma~\ref{lem_max_flow} to find if $G'$ can be transformed into a bipartite permutation graph by using at most $k'$ vertex deletions.
It either finds a \yes-instance or concludes after checking all the instances that there is no solution; that is, the initial instance is a \no-instance.

We note that such a branching into a bounded number of smaller instances is a standard technique, see e.g.,~\cite{HV13} for more details.

We now analyze the running time of the whole algorithm.
In the first part, observe that the branching tree has depth at most $k$ and has at most ${9^k}$ leaves, as $k$ decreases by one whenever the algorithm branches and each of the listed forbidden subgraphs has at most nine vertices. 
Therefore the total number of nodes in the branching tree is $\Oh{9^k}$. 
Moreover, in each node $(G'',k'')$, the algorithm works in time $\Oh{n^9}$ as it checks if $G''$ contains a forbidden set. 
In the second part, the algorithm does a work $\Oh{n^6}$ in each leaf, by Lemma~\ref{lem_max_flow}.
We conclude that the total running time of our algorithm for \bpd is $\Oh{9^k\cdot n^9}$.

\section{Proof of Theorem~\ref{thm_main_apx}}\label{sec_apx}

In this section, we provide a proof of Theorem \ref{thm_main_apx}. The idea of the algorithm is very similar to the \FPT algorithm described in Section \ref{sec_fpt_proof}.

Let $G=(V,E)$ be a graph and let $Y\subseteq V$ be a subset of vertices of $G$ such that $G-Y$ is a bipartite permutation graph. We want to construct a set $Z\subseteq V$ in polynomial time such that $G-Z$ is a bipartite permutation graph and $|Z|\leq 9|Y|$. We construct $Z$ as follows. We start with $Z=\emptyset$. Then, as long as $G-Z$ contains a set $X$ isomorphic to one of $K_3, T_2, X_2, X_3, C_5, C_6, C_7, C_8, C_9$ we add all vertices of $X$ to $Z$. Observe that $Y\cap X \neq \emptyset$ and $|X|\leq 9$.

After this step $G-Z$ is an almost bipartite permutation graph. Note that $|Z|\leq 9|Z\cap Y|$. We find a shortest hole $C=\{c_0, \dots, c_{m-1}\}$ in $G-Z$ and find a minimum hole cut $X$ as described in Section \ref{sec_fpt_proof}. Since $(Y-Z)$ is a hole cut in $G-Z$ we have $|X|\leq |Y-Z|$. We add $X$ to $Z$. Observe that $G-Z$ is a bipartite permutation graph.

Since $K_3, T_2, X_2, X_3, C_5, C_6, C_7, C_8, C_9$ have at most 9 vertices, we have that $|Z|\leq 9|Y|$. This implies that the above algorithm is a $9$-approximation algorithm. It runs in polynomial time because finding small forbidden subgraphs can be done in polynomial time and finding minimum hole cut in an almost bipartite permutation graph can be done in polynomial time.

\section{Conclusion}\label{sec_conclusion}

In this paper we investigate for the first time the modification problems in graph classes related to partial orders. 
Our main result says that the bipartite permutation vertex deletion problem is fixed parameter tractable.
We leave open the following two questions that inspired our research.
\begin{problem}
\label{problem:two_main_problems}
What is the parameterized status of the vertex deletion problems to the class of permutation graphs and to the class of co-comparability graphs?
\end{problem}
We recall that, due to the result of Lewis and Yannakakis \cite{LewYan78}, both of these problems are \NP-complete.
One of the most important result of our work is the description of the structure 
of almost bipartite permutation graphs, which are defined as graphs which do not induce small graphs from the list of forbidden structures for bipartite permutation graphs.
In a similar fashion we can define the class of \emph{almost permutation} and \emph{almost co-comparability graphs}.
The next two questions seem very natural in order to solve Problem~\ref{problem:two_main_problems}.
\begin{problem}
What is the structure of almost permutation and almost co-compara\-bility graphs?
\end{problem}
We are aware that the two problems mentioned above can be quite difficult.
Therefore, it is worth considering intermediate problems that may be easier to attack.
One of the proposed simplifications relies on the transition from the world of graphs to the world of posets.
The following \emph{vertex deletion into two-dimensional posets} problem seems very natural in the context of our research: we are given in the input a poset $P$ and a number $k$ and we ask whether we can delete at most $k$ points from $P$ so that the remaining points induce a two-dimensional poset in $P$.
\begin{problem}
\label{prob:vertex_deletion_into_two_dimensional_posets}
What is the parameterized status of the vertex deletion into two-dimensional poset problem?
\end{problem}
Since permutation graphs are co-comparability graphs of two-dimensional posets and since permutation graphs are both comparability and co-comparability graphs, 
the vertex deletion into two-dimensional poset problem is equivalent to the vertex deletion into co-comparability graph (or into permutation graph) problem if we assume that only comparability graphs can be given in the input.
The class of two-dimensional posets is very well understood; in particular, the list of minimal forbidden structures for this class of posets, which is still infinite, is known (obtained independently by Trotter and Moore \cite{TrotterMoore76} and by Kelly \cite{Kelly77}).
Of course, it is natural to ask the following question:
\begin{problem}
What is the structure of almost two-dimensional posets?
\end{problem}
Since the comparability graphs of posets do not contain odd holes of size $\geq 5$,
we know the structure of almost two-dimensional posets that are bipartite.
Indeed, these are the posets whose comparability graphs are almost bipartite permutation graphs embeddable into cylinder stripes.
The last problem we want to ask is as follows:
\begin{problem}
 Is there a polynomial kernel for the bipartite permutation vertex deletion problem?
\end{problem}
A positive answer to this question obtained by indicating so-called \emph{irrelevant vertices} may give some hope to solve Problem~\ref{problem:two_main_problems} with the use of irrelevant vertex technique.

\bigskip

\section*{Acknowledgment} The authors are grateful to Bartosz Walczak for valuable comments and help with merging two groups of researchers working on similar projects into one.
They also would like to thank the anonymous IPEC reviewers for their helpful comments.

\newpage
\bibliographystyle{plain}
\bibliography{references}

\end{document}